\documentclass[10pt,conference]{IEEEtran}
\IEEEoverridecommandlockouts
\newif\iftr % used to control whether this version is the technical report (= with appendix) or not
\trfalse
\if CLASSOPTIONcompsoc
\usepackage[nocompress]{cite}

\else
\usepackage{cite}

\usepackage{xcolor}
\def\BibTeX{{\rm B\kern-.05em{\sc i\kern-.025em b}\kern-.08em
		T\kern-.1667em\lower.7ex\hbox{E}\kern-.125emX}}

\usepackage{amssymb,amsmath,amsthm}
\usepackage[utf8x]{inputenc}
\usepackage{listings}
\usepackage{color}
\usepackage{epsfig}
\usepackage{url}
\usepackage{amsmath}
\usepackage{cite}
\usepackage{graphicx}
\usepackage{color}
\usepackage{amssymb,amsmath,dsfont}
\usepackage{amsthm}
\usepackage{epstopdf}
\usepackage{lipsum}
\usepackage{romannum}
\usepackage[ruled,linesnumbered, noend]{algorithm2e}
\usepackage{pgfplots}
\pgfplotsset{compat=1.13}
\usepackage{tikz}
\usepackage{graphicx}
\usepackage{breqn}

\usepackage{graphics}

\newtheorem{lemma}{Lemma}

\newtheorem{theorem}{Theorem}

\newtheorem{sub-goal}{Sub-goal}

\usepackage{csquotes}
\usepackage{amsmath}
\usepackage{amssymb}
\usepackage{subcaption}
\usepackage{graphicx}
\usepackage{accents}
\usepackage{mathtools}
\usepackage[normalem]{ulem}
\usepackage{dirtytalk}
\usepackage{mathrsfs}
\usepackage{amsmath}

\graphicspath{figures/}

\usepackage{mathtools}

\usepackage{xspace}

\newcommand{\set}[1]{\left\{{#1}\right\}}

\newcommand{\lmax}{l^{\textit{max}}}
\newcommand{\isdef}{\ensuremath{\overset{\textit{def}}{=}}}

\usepackage[textwidth=1.8cm]{todonotes}
\newcommand{\cV}{\mathcal{V}}
\newcommand{\cE}{\mathcal{E}}

\renewcommand{\lmax}{l^{\scriptsize\textrm{max},v}}
\newcommand{\tot}{\textrm{tot}}
\renewcommand{\isdef}{\ensuremath{\overset{\scriptsize\textit{def}}{=}}}
\newcommand{\lp}{\ensuremath{\left(}}
\newcommand{\rp}{\ensuremath{\right)}}
\newcommand{\lb}{\ensuremath{\left[}}
\newcommand{\rb}{\ensuremath{\right]}}

\newcommand{\mand}{\mbox{ and }}

\newcommand{\Reals}{\mathbb{R}}

\newcommand{\rmin}{R^{{\scriptsize\textrm{min},v}}}
\newcommand{\rmax}{R^{{\scriptsize\textrm{max},v}}}

\newcommand{\tmin}{T^{{\scriptsize\textrm{min},v}}}
\newcommand{\tmax}{T^{{\scriptsize\textrm{max},v}}}

\newcommand{\rminc}{\Bar{R}^{{\scriptsize\textrm{min},v}}}
\newcommand{\rmaxc}{\Bar{R}^{{\scriptsize\textrm{max},v}}}

\newcommand{\tminc}{\Bar{T}^{{\scriptsize\textrm{min},v}}}
\newcommand{\tmaxc}{\Bar{T}^{{\scriptsize\textrm{max},v}}}

\newcommand{\oldservice}{\beta^{\textrm{\scriptsize old},v}}
\newcommand{\newservice}{\beta^{\textrm{\scriptsize new},v}}

\newcommand{\mdelta}{d^{\scriptsize\textrm{max},v}}

\newcommand{\gammacon}{\gamma_c^{\scriptsize\textrm{convex},v}}

\newcommand{\betacdm}{\beta_c^{\scriptsize\textrm{CDM},v}}
\newcommand{\betacdma}{\beta^{\scriptsize\textrm{CDM},v}}

\newcommand{\plpd}{{\textrm{PLP}}^{\scriptsize\textrm{delay}}}
\newcommand{\iplpd}{{\textrm{iPLP}}^{\scriptsize\textrm{delay}}}
\newcommand{\plpb}{{\textrm{PLP}}^{\scriptsize\textrm{backlog}}}
\newcommand{\plpcut}{\textrm{FP-PLP}}
\newcommand{\tfa}{{\textrm{GenericTFA}}}

\newcommand{\tvar}[2]{\mathbf{t}_{(#1,#2)}}
\newcommand{\xvar}[1]{\mathbf{x}_{#1}}

\newcommand{\fvar}[4]{\mathbf{Ft}^{#1,#2}_{#3,#4}}
\newcommand{\depth}{\textrm{dp}}
\newcommand{\suc}{\textrm{sc}}
 \newcommand{\pd}{\textrm{perNodeDelay}}

\newcommand{\drrb}{{\textrm{DRRservice}}^{\scriptsize\textrm{outputBurst}}}
\newcommand{\drri}{{\textrm{DRRservice}}^{\scriptsize\textrm{inputArrival}}}

\newcommand{\zcuti}{z^{\scriptsize\textrm{cut}}_c}
\newcommand{\zcutc}{z^{\scriptsize\textrm{cut}}_c}

\newcommand{\fcutc}{F^{\scriptsize\textrm{cut}}_c}
\newcommand{\zcut}{z^{\scriptsize\textrm{cut}}}

\newcommand{\ecutc}{E^{\scriptsize\textrm{cut}}_c}
\newcommand\sbullet[1][.5]{\mathbin{\vcenter{\hbox{\scalebox{#1}{$\bullet$}}}}}

\newcommand{\dend}{d^{{\scriptsize\textrm{e2e}}}}

\newcommand\copyrighttext{%
	\footnotesize This work has been submitted to the IEEE for possible publication.
}
\newcommand\copyrightnotice{%
	\begin{tikzpicture}[remember picture,overlay]
		\node[anchor=north,yshift=-10pt] at (current page.north) {\parbox{\dimexpr\textwidth-\fboxsep-\fboxrule\relax}{\copyrighttext}};
	\end{tikzpicture}%
}

\begin{document}
	\pagenumbering{arabic}
	\pagestyle{plain}
	%Dimensioning Buffer Size and Delay Bound in Time-Sensitive Networking
	\title{Worst-case Delay Analysis of Time-Sensitive Networks with Deficit Round-Robin}
	%affiliations

%	\author{
%		\IEEEauthorblockN{...\\
%			\IEEEauthorblockN{Author 1, Author 2, Author 3\\
%		}
%	\IEEEauthorblockA{\'Ecole Polytechnique F\'ed\'erale de Lausanne, Switzerland\\
%	$\{$firstname.lastname$\}$@epfl.ch}}
%	

%	\author{\IEEEauthorblockN{Author 1}
%	\IEEEauthorblockA{\textit{Affiliation 1}\\
%		City 1, Country 1 \\
%		email@address}
%	\and
%	\IEEEauthorblockN{Author 2}
%	\IEEEauthorblockA{\textit{Affiliation 2}\\
%		City 2, Country 2\\
%		email@address}
%}
	\author{\IEEEauthorblockN{Seyed Mohammadhossein Tabatabaee}
	\IEEEauthorblockA{\textit{EPFL}\\
		Lausanne, Switzerland \\
		hossein.tabatabaee@epfl.ch}
		
	\and
	\IEEEauthorblockN{Anne Bouillard}
	\IEEEauthorblockA{\textit{Huawei Technologies France}\\
		Paris, France \\
		anne.bouillard@huawei.com}
	
	\and
	\IEEEauthorblockN{Jean-Yves Le Boudec}
	\IEEEauthorblockA{\textit{EPFL}\\
		Lausanne, Switzerland \\
		jean-yves.leboudec@epfl.ch}
}

	\maketitle
	\copyrightnotice
%		\copyrightnotice
        %\Marc{I have not worked enough to be second author.}
	
	\IEEEpeerreviewmaketitle
	
	\setcounter{page}{1}
	% !TeX root = mainTON.tex
\begin{abstract}
In feed-forward time-sensitive networks with Deficit Round-Robin (DRR), worst-case delay bounds were obtained by combining Total Flow Analysis (TFA) with the strict service curve characterization of DRR by Tabatabaee et al. The latter is the best-known single server analysis of DRR, however the former is dominated by Polynomial-size Linear Programming (PLP), which improves the TFA bounds and stability region, but was never applied to DRR networks. We first perform the necessary adaptation of PLP to DRR by computing burstiness bounds per-class and per-output aggregate and by enabling PLP to support non-convex service curves. Second, we extend the methodology to support networks with cyclic dependencies: This raises further dependency loops, as, on one hand, DRR strict service curves rely on traffic characteristics inside the network, which comes as output of the network analysis, and on the other hand, TFA or PLP requires prior knowledge of the DRR service curves. This can be solved by iterative methods, however PLP itself requires making cuts, which imposes other levels of iteration, and it is not clear how to combine them. We propose a generic method, called PLP-DRR, for combining all the iterations sequentially or in parallel. We show that the obtained bounds are always valid even before convergence; furthermore, at convergence, the bounds are the same regardless of how the iterations are combined. This provides the best-known worst-case bounds for time-sensitive networks, with general topology, with DRR. We apply the method to an industrial network, where we find significant improvements compared to the state-of-the-art. 

	\end{abstract} 
	% !TeX root = mainTON.tex
\section{Introduction}
\label{sec:intro}
Deficit Round-Robin (DRR) \cite{DRR} is a scheduling algorithm that is often used in real-time systems or communication networks for scheduling tasks, or packets.
%, in real-time systems or communication networks. %It is a \bleu{form of fair queuing} %Weighted Round-Robin (WRR)
%that enables flows with variable packet lengths to fairly share the capacity. The capacity is shared among several clients or queues by giving each of them a quantum value and  by providing more service to those with larger quantum. DRR is widely used because it exhibits a low complexity, $O(1)$, provided that an allocated quantum is no smaller than the maximum packet size; and it can be implemented in very efficient ways, such as the Aliquem implementation \cite{Aliquem}.
With DRR, every queue is assigned  a \emph{quantum} that is a static number. Queues are visited in \emph{round-robin} manner; the service received at every visit, which is measured in bits for communication networks or in seconds for task processing systems,  is up to the quantum value. Tasks or packets have variable sizes and it might happen that, during a visit of the server, at least one task or packet, which cannot be served, remains in the queue; this is because the unused part of the quantum is positive but not large enough. In such a case, the unused part of the quantum (called the \emph{residual deficit}) is carried over to the next round.
DRR shares resources flexibly (the amount of service reserved for one queue is proportional to its quantum), and is efficient (when a queue is idle, the server capacity is available to other queues). As it has low complexity and very efficient implementations exist~\cite{Aliquem}, it is widely used.

% With DRR, every queue is associated with a static number, called \emph{quantum}. Queues are visited \emph{round-robin} (one after the other), and at every visit, receive service (measured in bits for communication networks, in seconds for task processing systems) up to the quantum value. Tasks or packets are of variable sizes and it may happen that, during one visit of the server, there remains at least one task or packet in the queue that cannot be served because the unused part of the quantum is positive but not large enough. In such a case, the unused part of the quantum (called the \emph{residual deficit}) is carried over to the next round.
% DRR shares resources flexibly (the amount of service reserved for one queue is proportional to its quantum) and efficiently (when a queue is idle, the server capacity is available to other queues). It is widely used as it has low complexity and very efficient implementations exist~\cite{Aliquem}.

DRR can be applied to time-sensitive networks  where  bounds on worst-case delays (not on average) are required. Here, flows with similar delay requirements can be mapped to the same class, every class corresponds to one DRR queue at every node, and packets inside one class are handled first in first out (FIFO). Also, %Furthermore, 
flows are limited at sources by arrival curve constraints, i.e., limits to the number of bits that can be sent over any time interval.  Finding delay bounds in a DRR network involves two steps: a single node analysis and a combination of nodes in a per-class network analysis. For the former, in a recent RTAS conference, the authors in~\cite{drr_rtas} %(extended version in \cite{drr_ton}) 
derive a strict service curve for DRR, i.e., a function that lower bounds the amount of service received by every DRR queue. Delay bounds are then derived by using network calculus formulas. This method captures the interference of competing DRR classes and, as of today, provides the best known worst-case delay bounds \cite{1043123,10.5555/923589,Lenzini_fullexploitation,boyer_NC_DRR,anne_drr}.  We call this method the \emph{DRR strict service curve.} 

For the latter step, per-class network analysis, Total Flow Analysis (TFA) \cite{bouillard_deterministic_2018} was used in \cite{drr_rtas}. TFA obtains delay bounds in FIFO networks and can be applied to per-class networks that are FIFO per class and where a service curve is known for every class at every node. When applying TFA to DRR networks, DRR strict service curves require the knowledge  of burstiness bounds of competing classes inside the network, which is an output of the network analysis of TFA. Conversely, TFA needs to know the strict service curves. The authors in \cite{drr_rtas} solve this problem by considering only feed-forward networks (in the application example, they constrain flow routes to avoid cyclic dependencies). However, cyclic dependencies are frequent in time-sensitive networks and cannot be ignored. Recent versions of TFA~\cite{ludo_cycle,sync_TFA} apply to networks with cyclic dependencies and can therefore be used: as a side-contribution, in Section~\ref{sec:initPhase}, we show how to apply TFA to DRR networks with cyclic dependencies, by developing, and proving the validity of, an iterative procedure, called TFA-DRR.

However, our main contribution goes well beyond TFA-DRR; indeed, it is known that TFA is outperformed by polynomial-size linear programming (PLP)~\cite{plp}, which always provides delay bounds better than or equal to those of TFA, and, at high network utilization, often converges when TFA does not. Other methods such as LUDB \cite{ludb} and flow prolongations \cite{fp} also tend to dominate TFA, however unlike PLP, they do not apply to generic topologies. This motivates the goal of this paper, which is to design how PLP can be applied to DRR networks. The existing PLP, like the recent versions of TFA, applies to FIFO networks with any topology.  PLP consists in three phases: First, per-node delay bounds are computed (from TFA); Second, cuts are performed on the network topology in order to obtain a collection of trees and valid burstiness bounds are computed at the cuts by solving a linear program; Third, delay bounds for the flows of interest are computed on the cut network by solving another linear program for every flow of interest. PLP uses the per-node delay bounds obtained by TFA as a constraint in all its linear programs; it follows that the PLP delay bounds are guaranteed to be as good as the bounds obtained with TFA. An intriguing feature is that this enables PLP to obtain end-to-end delay bounds that are generally better than with TFA, whereas not using the per-node delay bounds as constraints may provide worse results.

Using PLP to analyse DRR networks requires to introduce the DRR strict service curve into the PLP procedure. PLP uses internal variables such as the burstiness bounds at cuts and the per-node delay bounds, the computation of which depend on the DRR strict service curves; the DRR strict service curves depend on burstiness bounds of interfering flows at the output of every node, which can be obtained by adding to PLP another family of linear programs;  the outputs of such linear programs depend on the burstiness at cuts, the per-node delay bounds, 
%\anne{called per-node in the peper. Should we unify?}\jylb{yes}, 
and the DRR strict service curves. In total, there are four collections of variables (burstiness bounds at cuts, burstiness bounds for interfering flows at DRR nodes, per-node delay bounds and DRR strict service curves), and the computation of every collection depends on the values of the other collections. It is natural to propose an iterative procedure as we mentioned above for the application of TFA to DRR (where there were only two collections, per-node delay bounds and DRR strict service curves), however it is not clear how the iterations should be combined and whether some specific combinations provide better bounds. To solve this issue, we propose a generic method to combine updates to any item in the four collections in any arbitrary order, using a distributed, shared-memory computing model. We show that the resulting bounds do not depend on how the item updates are executed, as long as every update is executed infinitely often in a hypothetical execution of infinite duration (Theorem~\ref{thm:genMethod}). Still, some concrete implementations of the method may have better execution times, and we propose two such concrete, parallel implementation methods, which we apply to the industrial network in \cite{drr_rtas}.

When applying PLP to DRR, we make two further improvements. First, the existing PLP obtains burstiness bounds for individual flows, whereas the DRR strict service curve uses burstiness bounds for the aggregate of all flows for every 
interfering DRR class at node output. Of course, a burstiness bound for an aggregate can be obtained by summing the burstiness bounds of every individual flow, but this is generally sub-optimal. In Theorem~\ref{thm:plpBack}, we extend the PLP methodology to obtain such per-aggregate bounds. Second, PLP requires convex service curves; \cite{drr_rtas} provides both convex and non-convex DRR strict service curves, and the latter may obtain smaller delay bounds when the delay bounds are very small. We solve this issue with a modification of PLP, called iPLP, that adds one binary variable to the linear program per DRR node (Section~\ref{sec:iplp}). 

The contributions of this paper are as follows:

$\sbullet[.75]$ We provide a method (PLP-DRR), for the worst-case timing analysis of per-class DRR network with or without cyclic dependencies, which combines DRR strict service curves and PLP in a novel way. It has three phases: (i) initial, which obtains initial TFA bounds; (ii) refinement, which improves the four collections of burstiness bounds at cuts, burstiness bounds for interfering flows at DRR nodes, per-node delay bounds and DRR strict service curves; (iii) post-process, which obtains delay bounds for flows of interest using iPLP. 
   
$\sbullet[.75]$ The refinement phase 
    uses a distributed computing model with shared memory, where individual improvements can be applied in any order. We show that any execution provides the same bounds, regardless of the order in which the individual improvements are applied. We prove that the bounds are valid. The bounds are guaranteed to be at least as good as the ones obtained with TFA-DRR.
    
$\sbullet[.75]$    
    We develop, and show the validity of, a method to apply TFA to DRR networks with cyclic dependencies. This method is of independent interest and is also used in the initial phase.
    
$\sbullet[.75]$    
    We design two improvements to the PLP methodology. The former computes improved burstiness bounds for aggregates of flows and is used in the refinement phase. The latter enables us to use non-convex service curves in PLP and is used in the post-process phase. 
    
$\sbullet[.75]$
        We design two concrete implementations of the method, with parallel for-loops in the refinement phase, and apply them to the industrial network in~\cite{drr_rtas}. We find that the delay bounds are significantly better than the state-of-the-art.
%\begin{itemize}
 %
  %   
 %   \item 
  %  \item We design two concrete implementations of the method, with parallel for-loops in the refinement phase, and apply them to the industrial network in~\cite{drr_rtas}. We find that the delay bounds are significantly better than the state of the art. %Also our method finds finite bounds in the high utilization region where existing methods do not.% to obtain finite bounds.
%\end{itemize}

The rest of the paper is organized as follows. Section~\ref{sec:sysmodel} describes the system model, including DRR operation, the network under study and the resulting graphs. Section~\ref{sec:bg} gives the necessary background on DRR strict service curve, TFA and PLP.  Section~\ref{sec:general_idea} gives a global view of PLP-DRR, our method for combining the DRR strict service curve and PLP. It uses two improvements of PLP, which are described and proven in Section~\ref{sec:2imp}. The details of PLP-DRR are described in Section~\ref{sec:detail}, including statements about the validity and the uniqueness of the obtained bounds. Section~\ref{sec:proof} contains proofs of theorems. Section~\ref{sec:numEval} applies the method to the industrial network in~\cite{drr_rtas} and illustrates the obtained improvements on delay bounds. Section~\ref{sec:conc} concludes the paper.

\section{System Model} \label{sec:sysmodel}

 We are interested in computing end-to-end delay bounds of flows in an asynchronous, time-sensitive packet-switched network with DRR. 
 % \bleu{In the context of Time-Sensitive Networking (TSN), DRR is one of the possible schedulers that is allowed.}

% in the context of deterministic networking. %In this section, we describe the model of a DRR system. 

%We consider a packet-switched network that uses DRR in the context of deterministic networking, and we are interested in the worst-case delays for flows, given the arrival curve constraints of flows at the source. 

%  We consider a packet-switched network that uses DRR, as explained in section \ref{sec:drr}, in the context of deterministic networking, and we are interested in the worst-case delays for flows, given the arrival curve constraints of flows at the source. Here we present our assumptions and system models in detail.
%  We consider a packet-switched network. We assume that there are several classes of traffic , and  flows are statically assigned to a class. At every node, packets of flows of different classes are processed according to DRR scheduling, as explained in \ref{sec:drr}. Also, within a class, packets of all flows are processed  are 
%   processed FIFO-per-class, i.e., in order of arrival.

\subsection{Deficit Round-Robin Scheduling}
\label{sec:drr}
%\subsection{Network Calculus \cite{le_boudec_network_2001, Changbook,bouillard_deterministic_2018}}
% !TeX root = mainDRR.tex

A DRR subsystem serves $n$ inputs, has one queue per input. %When a packet of input $i$ enters the deficit round-robin subsystem, it is put into queue $i$.
Each queue $c$ is assigned a quantum $Q_c$.
DRR runs an infinite loop of \emph{rounds}. In one round, if queue $c$ is non-empty, a service for this queue starts and its  \emph{deficit} is increased by $Q_c$. The service ends when either the deficit is smaller than the size of the head-of-the-line packet or the queue becomes empty. In the latter case, the deficit is set to zero. 
%The only operation that has a non-null duration is the packet sending. 
Packet sending duration depends both on the packet size and the amount of service available for the entire DRR subsystem. 

The DRR subsystem is placed in a larger system and can compete with other queuing subsystems. Then, the service offered to the DRR subsystem might not be instantly available. This can be modelled by means of a strict service curve $B()$, i.e., a function such at least $B(\tau)$ bits of any DRR class are served during any period of time of duration $\tau$ where the DRR subsystem is backlogged. A frequently used strict service curve
is the rate-latency function with rate $R$ and latency $T$, defined by $
\beta_{R,T}(t) = R[t-T]^+$, where we use the notation $[x]^+=\max\set{x,0}$.
%If the DRR subsystem has exclusive access to a transmission line of rate $R$, then $B(t)=Rt$ for $t\geq 0$. 
For example, when the DRR subsystem is at the highest priority on a non-preemptive server with line rate $R$, this gives a rate-latency strict service curve with rate $R$ and latency $\frac{R}{L^{\max}}$ where ${L^{\max}}$ is the maximum packet size of lower priority.
%\anne{the latency is 0 if the DRR sybsystem has exclusive access t the transmission line) (to keep??)}\jylb{no, this corresponds to $L^{\max}=0$}. 
If the DRR subsystem is not at the highest priority level, this can be modelled with a more complex strict service curve \cite[Section 8.3.2]{bouillard_deterministic_2018}.

Here we use the language of communication networks, yet the results equally apply to real-time systems: Simply map flow to task,  packet to job,  packet size to job-execution time, and  strict service curve to ``delivery curve" \cite{4617308,858698}.

\begin{figure}[htbp]
	\centering
	\includegraphics[width=0.72\linewidth]{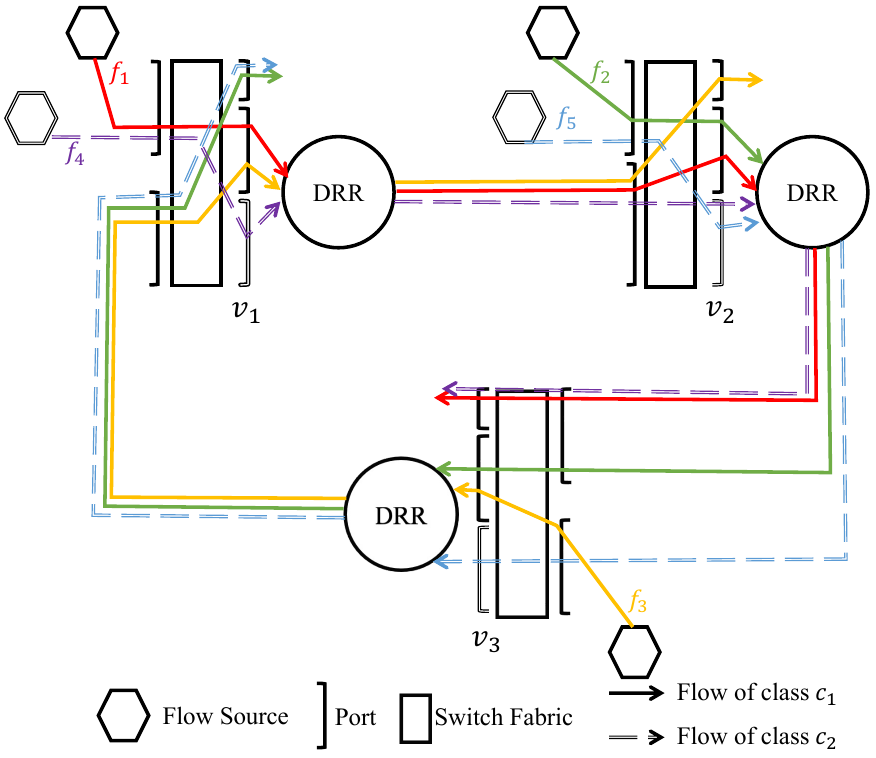}
	\caption{\sffamily \small Toy Network with 2 DRR classes. Flows $f_1$, $f_2$, and $f_3$ belong to class $c_1$; flow $f_4$ and $f_5$ belong to class $c_2$.}
	\label{fig:toyNet}
\end{figure}    
% \vspace{-0.5cm}
\begin{figure}[htbp]
	\centering
	%  \title={.}
	\begin{subfigure}[b]{0.23\textwidth}
		\centering
		% include first image
		\includegraphics[width=0.72\linewidth]{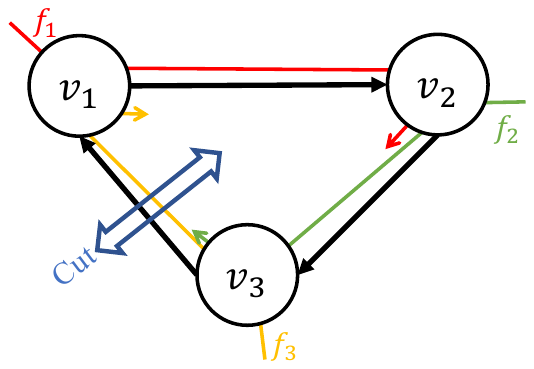}
		\caption{$\mathcal{G}_{c_1}$	}
	\end{subfigure}
	\hfill
	\begin{subfigure}[b]{0.23\textwidth}
		\centering
		% include first image
		\includegraphics[width=0.72\linewidth]{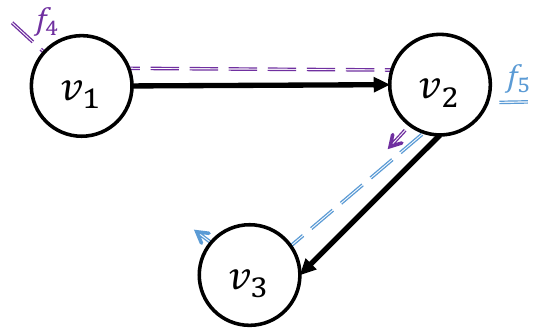}
	\caption{$\mathcal{G}_{c_2}$	 	}
	\end{subfigure}
	\caption{\sffamily \small The graphs induced by flows of class $c_1$ and $c_2$ of toy network of Fig.~\ref{fig:toyNet}, also showing the flows path. $\mathcal{G}_{c_1}$ has one cyclic dependency, $\mathcal{G}_{c_2}$ has none.}
	\label{fig:toyG}
\end{figure}

\subsection{Network Model and Resulting Graphs}
% We now describe and give the model for the the complete system.
\subsubsection{Device Model}
Devices represent switches or routers and consist of input ports, output ports, and a switching fabric. Each packet enters a device via an input port and is stored in a packetizer. A packetizer releases a packet only when the entire packet is received. Then, the packet goes through a switching fabric. A switching fabric transmits the packet to a specific output port, based on the static route of the packet. Then, the packet, based on its static class, is either queued in a FIFO-per-class queue or exits the network via a terminal port. At each non-terminal output port,  packets of flows of different classes are processed according to DRR scheduling, as explained in \ref{sec:drr}. See Fig.~\ref{fig:toyNet}. The aggregate service received by the DRR scheduler of a non-terminal output port $v\in\mathcal{V}$ is modeled by a strict service curve $B^v$, that we assume to be a rate-latency function $\beta_{R^v,T^v}$ with rate $R^v$ and a latency $T^v$ (not to be confused with the strict service curve offered by the DRR scheduler to each class).

\subsubsection{Flow Model}
  We assume that there are $n$ classes of traffic $1,\ldots, n$ in the system, and  flows are statically assigned a class and a path. Traffic generated by flows is constrained at the source. This can be modeled by means of an arrival curve $\alpha()$, i.e., a function such that the number of bits observed during any period of duration $\tau$ is at most $\alpha(\tau)$. Specifically,  each flow $f$ is constrained at source by a token-bucket arrival curve defined by $\gamma_{r_f,b_f}(t) = r_ft+b_f$ for $t>0$ and $\gamma_{r_f,b_f}(t)=0$ for $t= 0$. Having a token-bucket arrival curve with a rate $r_f$ and a burst $b_f$ means that the number of bits observed for this flow in any time interval of duration $\tau$ cannot be more than $b_f + r_f\tau$ bits.

\subsubsection{Graph induced by flows}
For every class $c$, the graph $\mathcal{G}_c = (\cV_c,\mathcal{E}_c)$ induced by flows is the directed graph defined as follows: 1) $\cV_c \subseteq \cV$ is the subset of all non-terminal output ports used by at least one flow of class $c$. 2) The directed edge $e = (v,u) \in \cE_c$ exists if there is at least  one flow of class $c$ that traverses $v$ and $u$ in this order. We say that $\mathcal{G}_c$ has a cyclic dependency if it contains at least one cycle. Let $E^{\scriptsize\textrm{cut}}_c \subseteq \mathcal{E}_c$ be a cut such that artificially removing the edges in $E^{\scriptsize\textrm{cut}}_c$ creates a tree or a forest (i.e., a collection of non-connected trees). Such cuts can be obtained by any traversal graph algorithm~\cite{cormen}.
%Edmond's algorithm \cite{edmonds1967optimum} with unit weights that provides minimum number of cuts. 
We keep the same node naming of vertices across graphs of different classes, namely, output ports of different classes that are connected to the same DRR scheduler have the same name. Let $\textrm{In}_c(v) \subset \mathcal{E}_c$ (resp. $\textrm{Out}_c(v)$) denote the set of edges of class $c$ that are incidents  at (resp. leave) node $v$. Consider the toy network of Fig.~\ref{fig:toyNet}. We assume we have two classes where $f_1$, $f_2$, and $f_3$ belong to class $c_1$ and, $f_4$ and $f_5$ belong to class $c_2$. The graph induced by flows of $c_1$ and $c_2$ as well as the flow paths are illustrated in Fig. \ref{fig:toyG}. Graph $\mathcal{G}_{c_1}$ has a cycle; the figure shows one possible artificial cut to create a tree.

% \anne{Question: What i the actual topology in the paper: do we consider the set of tree? the original topology? (In particular when computing $PLP_{backlog}$ at $Out(v)$...)}
% \hossein{My answer: If we separate $\plpb$, 1) we need to add some new variables in the shared memory: aggregate burstiness for each class at every node. 2) these variables are only used by $\drrb$. That is why I made them local variables such that whenever $\drrb$ is visited, we compute them. This contrasts other variables of the shared memory because they are used in every refinements.}

% The graph $\mathcal{G} = (\mathcal{V},\mathcal{E})$ induced by flows is a directed graph where (1) a node $v \in \cV$ is an output port, and (2) a directed edge $e = (v,u) \in \cE$ exists if there is at least  one flow that traverse $v$ and $u$ consecutively. For each class $i$, the graph induced by flows of this class is $\mathcal{G}_i = (\cV_i,\mathcal{E}_i)$ 
% %induced by flows of class $C_i$.  For each class $i$, graph $\mathcal{G}_i$ 
% and has a general topology. Let $E^{\scriptsize\textrm{cut}}_i \subseteq \mathcal{E}_i$ be a cut such that artificially removing them creates a tree or a forest (i.e., a collection of trees).
% \anne{the fonts are changing $E_i$or $\mathcal{E}_i$?}

% \subsection{Notation List}

% \usepackage{graphicx}
\begin{table*}[]
\centering
	\caption{\sffamily \small Notation List}
\resizebox{1.5\columnwidth}{!}{%
\begin{tabular}{|l|l|}
\hline
$\cV, v$               & The set of all output ports, an output port                                          \\
$\mathcal{G}_c = (\mathcal{V}_c,\mathcal{E}_c)$               & The graph induced by flows of class $c$                                            \\
$\mathcal{V}_c$                                           & The set of all output ports of class $c$                                          \\      
%$v \in \mathcal{V}_c$                                     & An output port                                                          \\
$\mathcal{E}_c$, $e$                                          & The set of edges of class $c$, an edge of class $c$                                                      \\
%$e \in \mathcal{E}_c$                                     & An edge in network                                                      \\
$E^{\scriptsize\textrm{cut}}_c$                         & Cutset: removing $E^{\scriptsize\textrm{cut}}_c$ creates a tree or a forest for class $c$ \\
$B^v$                                 & Aggregate strict service curve offered to $v$            \\
$f$, $\alpha_f = \gamma_{r_f,b_f}$                                                    & A flow, its token-bucket arrival curve                                                     \\
%$\alpha_f = \gamma_{r_f,b_f}$                           & Flow $f$  constrained by a token-bucket arrival curve                   \\ \hline
$\textrm{In}_c(v) \subset \mathcal{E}_c$ & The set of edges of class $c$ that are incidents at node $v$                        \\
$\textrm{Out}_c(v) \subset \mathcal{E}_c$          & The set of edges of class $c$ that leave node $v$                                   \\
\hline
$z_c^e$ &
Collection of burstiness upper bounds for transit flows of class $c$ carried by edge $e$ \\
 % \begin{tabular}[c]{@{}l@{}}A vector that has an upper bound on the propagated burstiness \\ for every transit flow of class $c$ that is carried out by edge $e$\end{tabular} \\
  $z_c^E$             & Collection of $z_c^e$ such that $e \in E$               \\
$z$                                                     & Collection of $z_c^{\mathcal{E}_c}$ of all classes $c$                    \\
%$E^{\scriptsize\textrm{cut}}_c$                         &Cuts that removing them create trees for class $c$                                                                        \\
$\zcuti$ &
  %\begin{tabular}[c]{@{}l@{}}
  Upper bounds on the burstiness of %\\ 
  flows of class $c$ at cuts  $E^{\scriptsize\textrm{cut}}_c$ 
  %\end{tabular} 
  \\
$\zcut$                       & Collection of $\zcuti$ for all classes $c$ \\
$d_c^v$                                                 & Delay bound on node $v$ for class $c$                                 \\
$d_c$                                                   & Collection of delay bounds at all nodes for class $c$                       \\
$d$                                                     & Collection of $d_c$ (per-node, per-class delay bounds)                           \\
$\beta_c^v$                                             & Strict service curve offered to class $c$ at node $v$                 \\
$\beta_c$                                               & Collection of per-node strict service curve offered to class $c$ at all nodes               \\
$\beta$                                                 & Collection of  $\beta_c$  (per-node, per-class strict service curves)     \\
$b_c^v$ & bound on aggregate burstiness of flows of class $c$ that exits node $v$\\
$b_c^e$ & bound on aggregate burstiness of flows of class $c$ carried by edge $e$\\
$b$ & Collection of $b_c^v$ and $b_c^e$ for every class $c$, every edge $e$, and every node $v$  \\\hline
$\beta_{R,T}$ &  $\beta_{R,T}(t) = R[t-T]^+$, rate-latency function\\
$\betacdm$ & DRR strict service curve, no assumption on arrival curves\\
$\beta_c^{\scriptsize\textrm{nc},v}$ & Non-convex DRR strict service curve
\\ \hline
% % $\tfa_c $ & todo \\
% % $\drri\lp\alpha^v\rp $ & todo \\
% % $\lp \beta^v \rp = \drrb\lp \beta^v, b^v \rp $ & todo \\
% % $\lp \beta^v \rp = \drrb\lp \beta^v, b^v \rp $ & todo \\
\end{tabular}%
}\label{tab:not}
\end{table*}

\section{Network Calculus Background}
\label{sec:bg}
In this section, we provide the necessary background for analyzing a DRR system. In the network calculus framework, the classical method for this analysis combines two techniques: 1) the computation of strict service curves for each DRR class at each node, presented in Section~\ref{sec:drrservice}; 2) the analysis of one FIFO network per DRR class. Two methods are presented, TFA in Section~\ref{sec:tfa} and PLP in Section~\ref{sec:plp}.

\subsection{Strict Service Curves of DRR} \label{sec:drrservice}
%\jylb{We say convex and then we spend most of the section discussing a non convex sc !}
Authors in \cite{drr_ton} (extended version of \cite{drr_rtas}) find DRR strict service curves for degraded operational mode (i.e., where no assumption on the arrival curve of the interfering traffic is assumed) and for non-degraded operational mode (i.e., where some arrival curves can be assumed for the interfering traffic). They show that their 
% strict service curves and the 
resulting delay bounds dominate all previous works for single node analysis \cite{1043123,10.5555/923589,Lenzini_fullexploitation,boyer_NC_DRR,anne_drr}, and hence, are the best, proven ones.

%Here we summarize their convex strict service curves and one non-convex part of their service curves.
%, and we present a corollary to refine strict service curves of DRR using the output bursts.

\subsubsection{Degraded Operational Mode}
%: No Assumption on the Arrival Curves of Interfering Traffic} 
\label{sec:drrservice_dm}
%\subsection{Non-Convex Strict Service Curves}
% \begin{theorem}[Theorem 1 of \cite{drr_ton}: Non-convex Strict Service Curve for DRR]\label{thm:drrService}
% 	Let $S$ be a server that, shared by $n$ flows,  uses DRR, as explained in Section \ref{sec:drr}, with quantum $Q_i$ for class $i$. The server offers a strict service curve $B$ to the aggregate of the $n$ classes. For any class $i$, $\mdelta_i$ is the maximum residual deficit.
	
% 	Then, for every $i$, $S$ offers to class $i$ a strict service curve $\betancdm$ given by $\betancdm(t)=\gamma_i \lp \beta(t) \rp$ with
% 	\begin{align}
% 		\label{eq:gamma}
% 		\gamma_i (x) &= \lp \lambda_1 \otimes \nu_{Q_i,Q_{\tot}} \rp \lp  \lb x - \psi_i \lp Q_i - \mdelta_i \rp \rb^+ \rp \\ \nonumber &+ \min  ([x - \sum_{j \neq i} \lp Q_j + \mdelta_j\rp  ] ^+ , Q_i - \mdelta_i  )
% 		%		\gamma_i (x) &= \lp \lambda_1 \otimes \nu_{Q_i,Q_{\tot}} \rp \lp  \lb x - F_i \rb^+ \rp \\ \nonumber &+ \min  ([x - \sum_{j,j \neq i} Q_j + \mdelta_j  ] ^+ , Q_i - \mdelta_i  )
% 		\\
% 		\label{eq:Qtot}
% 		Q_{\tot} &= \sum_{j=1}^n Q_j
% 		%	\\
% 		%	F_i &= Q_i - \mdelta_i + \sum_{j,j \neq i}\phi_{i,j}(Q_i - \mdelta_i )
% 		\\
% 		\label{eq:psi}
% 		\psi_i(x) &\isdef x + \sum_{j,j \neq i} \phi_{i,j} \lp x \rp
% 		\\
% 		\label{eq:phi}
% 		\phi_{i,j}(x) &\isdef \left\lfloor \frac{x + \mdelta_i}{Q_i} \right\rfloor Q_j + \lp  Q_j + \mdelta_j \rp
% 	\end{align}
	
% 	Here, $\nu_{a,b}$ is the stair function, $\lambda_1$ is the unit rate function and $\otimes$ is the min-plus convolution, all described in Fig.~\ref{fig:minplus}.
% \end{theorem}
Let $v$ be a node shared by $n$ classes that uses DRR (see Section~\ref{sec:drr}), with quantum $Q_c$ for class $c$. The node offers a strict service curve $B^v$ to the aggregate of the $n$ classes. 
%For any class $i$, $\mdelta_i$ is the maximum residual deficit. 
Then, for every class $c$, node $v$ offers to class $c$ a strict service curve $\betacdm$ that is the maximum of two rate-latency functions, and hence, is piece-wise and convex; The rate and latency depend on the quanta and maximum residual deficits (see Appendix \ref{app:drr} for details).

Non-convex strict service curves of DRR can improve delay bounds when they are small, specifically if the service for a flow finishes in the first round (i.e., the flow is never backlogged at the end of each of its round of service).
% the first time that this class is visited by DRR scheduler). 
This motivates us to consider only the first non-convex part of the DRR strict service curve, say $\beta_c^{\scriptsize\textrm{nc},v}$, as it corresponds to the first service round. 
% \jylb{Do not say: we denote with $x$, instead, say: we let $x$ denote...} 
%We let denote it by 
We have $\beta_c^{\scriptsize\textrm{nc},v}(t) = \min\lp \beta_{R^v,T_1}(t), q_c^v \rp$
%Q_i - \mdelta_i\rp$, 
where $T_1$ is the maximum period of time during which no data of class $c$ can be served and %$Q_i - \mdelta_i$ 
$q_c^v$ corresponds to the minimum amount of data guaranteed for class $c$ during each round
%$\beta_i^{\scriptsize\textrm{nc},v}$ where $\beta_i^{\scriptsize\textrm{nc},v}(t) = \min\lp \beta_{c_v,\tmin_i}(t), Q_i - \mdelta_i\rp$ 
(see Fig. \ref{fig:ncService}); the exact values are given in Appendix~\ref{app:drrDM}. We use it as follows in Section~\ref{sec:iplp}: since $\beta_c^{\scriptsize\textrm{nc},v}$ is a strict service curve, we can replace any other strict service curve for class $c$, $\beta_c^v$, by $\max(\beta_c^v, \beta_c^{\scriptsize\textrm{nc},v})$, which is also a strict service curve.  
% ; hence, maximum of $\beta_i^{\scriptsize\textrm{nc},v}$ and any other available strict service curve $\beta_i^v$ is also a strict service curve. 
%We use it as follows: 
%whenever $\beta_c^v$ is another strict service curve for class $c$, then we can replace it by $\max(\beta_c^v, \beta_c^{\scriptsize\textrm{nc},v})$, which is also a strict service curve.  
% \anne{hence, if $\beta_c^v$ is a strict service curve for class $c$, then $\max(\beta_c^v, \beta_c^{\scriptsize\textrm{nc},v})$ is also a strict service curve.  }
%\jylb{I don't understand what we want to say here. Are we replacing the convex service curve by this one ? It has $0$ rate beyond $T_2$ !} \hossein{I did pass on this part}
\begin{figure}%[htbp]
	\centering
\definecolor{mycolor1}{rgb}{0.00000,0.44700,0.74100}%
\definecolor{mycolor2}{rgb}{0.85000,0.32500,0.09800}%
\definecolor{mycolor3}{rgb}{0.92900,0.69400,0.12500}%
\begin{tikzpicture}
\begin{axis}[%
width=2.7in,
height=1.7in, %754in,
at={(0in,0in)},
% scale only axis,
xmin=0,
xmax=200,
xlabel style={font=\color{white!15!black}},
xlabel={Time ($\mu s$)},
ymin=0,
ymax=180,
ylabel style={font=\color{white!15!black}},
ylabel={Bits},
axis background/.style={fill=white},
axis x line=bottom,
axis y line=left,
legend style={legend cell align=left, align=left, draw=white!15!black},
        xtick = {50,75},
        xticklabels = {$T_1$,$T_2$},
                ytick = {100},
        yticklabels = {$q_c^v$},
% xticklabels={},
% yticklabels={},
% extra x ticks={45},
% extra x tick labels={$T_1$},
]

% \draw [-> , thick, color = blue] (50,70) -- (60,80)node[pos=0.5, above, , yshift=0.275cm] {slope $c$};
 \draw [<-> , dashed, color = blue] (0,35) -- (60,35)node[pos=0.5, above,  yshift=0.275cm] {$d^{\scriptsize\textrm{nc},v}_c$};
  \draw [<-> , dashed, color = red] (0,30) -- (85,30)node[pos=0.5, below,  yshift=0.05cm] {$d^{\scriptsize\textrm{c},v}_c$};
\draw [-> , thick, color = blue] (65,110) -- (70,100)node[pos=0.5, left] {$\beta^{\scriptsize\textrm{nc},v}$};
\draw [-> , thick, color = red] (130,120) -- (135,110)node[pos=0.5, left] {$\beta^{\scriptsize\textrm{c},v}$};
\draw [-> , thick, color = black] (170,60) -- (165,75)node[pos=0.5, right] {$\alpha^v_c$};

\draw [dotted, color = black] (75,100) -- (75,0);
\draw [ dotted, color = black] (75,100) -- (0,100);

% \draw (45,-5) node{$T_1$};
% \draw (85,-5) node{ \small{$T_2 = T_1 + \frac{q_c^v}{R^v}$}};
%\draw (35,110) node{\small $c_v(T_2 - T_1) = Q_i-\mdelta_i$};
% \draw (5,110) node{$q_c^v$};

\addplot [color=blue, line width=1.5pt ]
table[row sep=crcr]{%
    0      0\\
    50      0\\
    75      100\\
    200     100\\
};

\addplot [color=red, line width=1pt ]
table[row sep=crcr]{%
    0      0\\
    50      0\\
    90      30\\
    130    90\\
    200    250\\
};

\addplot [color=black, line width=1pt ]
table[row sep=crcr]{%
    0      0\\
    0      40\\
    200    90\\
};

\end{axis}
\end{tikzpicture}%
	\caption{\sffamily \small Non-convex part of a service curve $\beta^{\scriptsize\textrm{nc},v}(t) = \min\lp \beta_{R^v,T_1}(t),q_c^v\rp$ and $T_2 = T_1 + \frac{q_c^v }{R^v}$ at some node $v$ and for some class $c$. The convex function $\beta^{\scriptsize\textrm{c},v}$ is also a valid service curve, so is the maximum of $\beta^{\scriptsize\textrm{c},v}$ and $\beta^{\scriptsize\textrm{nc},v}$. The figure shows an arrival curve $\alpha_c^v$ to illustrate that the bound obtained with the non-convex service curve, $d^{\scriptsize\textrm{nc},v}_c$, is better than that of the convex service curve, $d^{\scriptsize\textrm{c},v}_c$,  when the delay bound is small. %Both Also, the figure illustrates  a convex strict service curve $\beta^{\scriptsize\textrm{nc},v}$ and an arrival curve $\alpha_c^v$ for this class. Then, the obtained bound using the non-convex service curve, $d^{\scriptsize\textrm{nc},v}_c$, is better than that of the convex service curve, $d^{\scriptsize\textrm{c},v}_c$,  when the delay bound is small.
 %\beta_{c_v,T_1}(T_2)\rp$.
	%, where $T_1 = T_v + \tmin_i$ and $T_2 = \tmin_i + \frac{Q_i-\mdelta_i}{c_v}$.
		% it also shows the fixed-point strict service curve $\beta_i^0^*$ of Theorem \ref{thm:optDrrService} by assuming $\alpha_j(t) =\lmax_j \lceil \frac{t}{512} \rceil$ for $j \neq i$ (note that $\beta_i^0^*$  is not a rate-latency function)
		%	\jylb{
		%the figure does not show the convex service curve	}
		%\ihossein{done}
	}
	\label{fig:ncService}
\end{figure}
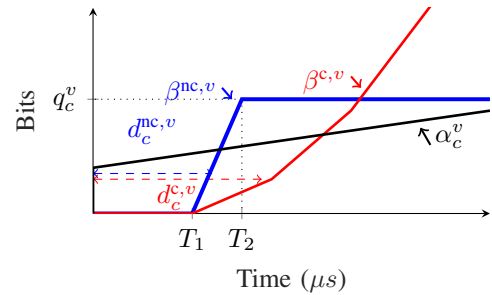

\subsubsection{Non-Degraded Operational Mode}\label{sec:drrservice_ndm}
%: Some Arrival Curves can be assumed for Interfering Traffic}
When some arrival curves can be assumed for the interfering classes, the service received by the class of interest can be improved. Authors in \cite{drr_ton} present a method that starts from service curves with no assumption on the interfering traffic (i.e., $\betacdm$ explained in Section \ref{sec:drrservice_dm}), and iteratively improves them by taking into account the arrival curve constraints of interfering traffic. 
% \hossein{I have removed the theorems}
% \jylb{I find it confusing to give it a new theorem number (Theorem 5) when in reality we are citing existing work.}
% \hossein{Modfied it.}
% See Theorem~5 and Algorithm~2 of \cite{drr_ton} for more details.  
We call 
%this method $\beta^v = \lp\beta_1^v,\ldots,\beta_n^v \rp = 
% $\drri_v\lp\alpha_1^v,\ldots,\alpha_n^v \rp$ 
$\drri_v\lp\alpha^v\rp$ 
the method that computes a collection of strict service curves for each class given $\alpha^v$, input arrival curves of all classes at node $v$ (See Appendix~\ref{app:drrNDM} for more details.).

We also use \eqref{eq:betamJ} in Appendix~\ref{app:drrNDM} (same as Lemma 2 in \cite{drr_ton}), which improves the strict service curves of DRR, given \emph{output}  arrival curves of every class. Specifically, we apply it when 
% We also use an application of Theorem \ref{thm:map} in Appendix \ref{app:drrNDM}, which improves the strict service curves of DRR, given \emph{output}  arrival curves of every class. \jylb{This is even more confusing. What is Theorem 6 ? A new result ? If so it should not appear in appendix. If it is not new, we should present it differently.} Specifically, we apply it when 
% We also use Lemma 2 in \cite{drr_ton}, which improves the strict service curves of DRR, given \emph{output}  arrival curves of every class. Specifically, we apply it when 
% This result can be simplified  in a specific case that is used later in this paper: Assume that
every class $c$ has a token-bucket arrival curve at the output, say $\gamma_{r_c,b_c^v}$, and a known strict service curve $\beta^v_c$. 
We call
%$\newservice  =  \drrb \lp b^*, \oldservice \rp$
$\drrb_v \lp \beta^v ,b^v\rp$ the function that implements \eqref{eq:betamJ} in Appendix~\ref{app:drrNDM}
% Lemma 2 in \cite{drr_ton}
and returns an improved collection of strict service curves for all classes at node $v$. 
% (See \eqref{eq:betamJ} in Appendix~\ref{app:drrNDM} for details). %\jylb{fix inconsistent notation with algo 1} 
In the above, $\beta^v$ and $b^v$ are the collection of $\beta_c^v$ strict service curves and collection of $b_c^v$ output burstiness of all classes at node $v$. 

%We call this operation 
%$\newservice  =  \drrb \lp b^*, \oldservice \rp$
%$\drrb \lp \beta^v ,b^v\rp$.
%
%Let $\beta^v$ and $b^v$ be the collection of $\beta_c^v$ strict service curves and collection of $b_c^v$ output burstiness of all classes at node $v$.
%$\oldservice_i$. % is known for class $C_i$.
%Lemma 2 in gives an improved service curve %$\newservice_i$ 
%for each class given output burstiness of other classes.
In the common case where the aggregate strict service curve is a rate-latency function say $\beta_{c_v,T_v}$, the obtained strict service curves are the maximum of a finite number of rate-latency functions.
% : the latency of such functions is a linear function of output burstinesses $b_c^v$, and the rates depend only on arrival rates $r_c$. 
 Specifically, consider  class  $c$ and partition other classes into two arbitrary sets $J$ and $\Bar{J}$. Then, let $r^{v,\bar{J}} =  \sum_{c' \in \bar{J}}r_{c'}$ and  $b^{v,\bar{J}} = \sum_{c' \in \bar{J}}b_{c'}$, i.e., $r^{v,\bar{J}} $ and $b^{v,\bar{J}}$ are the aggregated arrival rate and aggregated output burstiness bound of interfering classes in $\bar{J}$. Then, two rate-latency functions $\beta_{{\rmaxc_c }^J,{\tmaxc_c }^J} $ and $\beta_{{\rminc_c}^J,{\tminc_c}^J}$ can be computed for class $c$ with 
 \begin{align}
     {\rmaxc_c}^J &= (c_v - r^{v,\bar{J}}){\rmax_c }^J 
     \label{eq:ratemax}\\
     {\rminc_c}^J &= (c_v - r^{v,\bar{J}}){\rmin_c}^J
     \label{eq:ratemin}\\
          {\tmaxc_c }^J &=\frac{c_vT_v +{\tmax_c }^J + b^{v,\bar{J}} }{(c_v - r^{v,\bar{J}})} 
     \label{eq:latencymax}\\
     {\tminc_c }^J &=\frac{c_vT_v + {\tmin_c }^J + b^{v,\bar{J}} }{(c_v - r^{v,\bar{J}})}  
     \label{eq:latencymin}
 \end{align}
 Values of ${\rmax_c }^J$, ${\tmax_c }^J$, ${\rmin_c}^J$, and ${\tmin_c }^J$ depend on the quantum and the maximum residual deficits  of class $c$ and interfering classes in $J$; the exact values are given in Appendix \ref{app:drrNDM}. Thus, any choices of sets $J$ and $\Bar{J}$ results in two rate-latency functions for class $c$; function $\drrb_v \lp \beta^v ,b^v\rp$ computes the maximum of rate-latency functions obtained by all possible choices of sets $J$ and $\Bar{J}$. Observe that the latencies of the above functions, ${\tmaxc_c }^J$ and ${\tminc_c }^J$  in \eqref{eq:latencymax} and \eqref{eq:latencymin}, are  linear functions of output burstinesses of interfering classes $b_c^v$, and the rates, ${\rmaxc_c}^J$ and ${\rminc_c}^J$  in \eqref{eq:ratemax} and \eqref{eq:ratemin}, only depend  on arrival rates $r_c$.
 % ${\rmax_c}^J = \frac{(c_v - r^{v,\bar{J}})Q^v_c}{Q_c  + \sum_{c'\in J} Q_{c'} }$ and ${\rmin_c}^J &= \frac{((c_v - \sum_{c' \in \Bar{J}}r_{c'}))Q^v_c - \mdelta_c}{Q^{J,c}_\tot - \mdelta_c}$

 % let $N_c = \{c_1,c_2,\ldots,c_n\}\setminus\{c\}$, and consider a $J \subseteq N_c$ and let $\bar{J} =  N_c \setminus J$. Then, two rate-latency functions $\beta_{{\rmax_c }^J,{\tmax_c }^J} $ and $\beta_{{\rmin_c}^J,{\tmin_c}^J}\rp$ can be computed for class $c$ with ${\rmax_c}^J = \frac{(c_v - \sum_{c' \in \Bar{J}}r_{c'})Q^v_c}{Q_c  + \sum_{c'\in J} Q_{c'} }$ and ${\rmin_c}^J &= \frac{((c_v - \sum_{c' \in \Bar{J}}r_{c'}))Q^v_c - \mdelta_c}{Q^{J,c}_\tot - \mdelta_c}$}
% \begin{align}
%     \beta_{{\rmax_c }^J,{\tmax_c }^J} 
% \end{align}

	% \begin{align}
	% 	\label{eq:gammaMaxrateJ}
	% 			{\gammacon}^J &= 	 \max\lp\beta_{{\rmax_c }^J,{\tmax_c }^J} , \beta_{{\rmin_c}^J,{\tmin_c}^J}\rp\\
	% {\rmax_c }^J&= \frac{Q^v_c}{Q^{J,c}_\tot},~ {\tmax_c }^J= \sum_{c' \in J}\lp Q^v_{c'} + \mdelta_{c'}   +\frac{Q^v_{c'}}{Q^v_c} \mdelta_c  \rp\\
	% {\rmin_c}^J &= \frac{Q^v_c - \mdelta_c}{Q^{J,c}_\tot - \mdelta_c},~  {\tmin_c }^J= \sum_{c' \in J}\lp Q^v_{c'}  + \mdelta_{c'}  \rp		\\
	% 		\label{eq:QtotJ}
	% 	Q_{\tot}^{J,c} &=Q_c  + \sum_{c'\in J} Q_{c'} 
	% \end{align}
%\jylb{What is $\drrb$ ? where defined ? Is it new ? Is it from \cite{drr_ton} ? }
\subsection{Total Flow Analysis (TFA)} \label{sec:tfa}

Total Flow Analysis (TFA) \cite{tfa,tfa++,sync_TFA} is a method to conduct worst-case analysis in a FIFO network. In a per-class network where a service curve is known for every class at every node, one instance of TFA is run per class, and it outputs per-node  delay bounds as well as propagated burstiness for flows. %TFA is simple yet it can consider several important features such as the effect of packetizer and line-shaping. I
%It needs to know the arrival curves of flows at sources and offered service curves to each node, and it can be applied to a network with general topology. 
If the graph induced by flows is feed-forward (i.e, cycle-free), for each node in a topological order, a delay bound and output burstiness bounds of flows are computed: the output burstiness bounds at a node are used as input by its successors in the induced graph. 
%TFA first computes network calculus delay bound and output burstiness bounds of flows at edge-nodes (such nodes exists as the graph is cycle free). The output burstiness are then used as input for the following nodes. 
%Thus, TFA iteratively computes delay bounds for all nodes and burstiness bounds for all edges of the network. 
Else if the graph induced by flows has cyclic dependencies, no topological order can be defined  and % edge-nodes do not exist and 
a fixed point must be computed, using an iterative method \cite{sync_TFA}. If the iteration converges to a finite value for all delay and burstiness bounds, % bounds for all nodes and burstiness bounds for all edges of the network, 
then the network is stable and the computed bounds are valid. Otherwise, TFA diverges and the network might be truly unstable or not. 

All versions of TFA (specifically, FPTFA, SyncTFA, AsyncTFA, and AltTFA) are equivalent, i.e., they give the same bounds and stability regions \cite{sync_TFA}. We let  $(d_c,z_c) = \tfa_c \lp\beta_c\rp $ denote any version of TFA that computes per-node delay bounds and bounds on propagated burstiness for class $c$, given per-node  strict service curves $\beta_c$. We always apply TFA to the original (uncut) graph.

In networks with cyclic dependencies, there is a two-way dependency between TFA and DRR strict service curves: TFA needs to know DRR strict service curves a priori, however, DRR strict service curves depend on burstiness bounds of flows at the output of a node, which is a result of TFA. Specifically, on the one hand, DRR strict service curves of node $v$ depend on $b^v$,  the collection of $b_c^v$ output burstiness bound of all classes at node $v$ (see Section \ref{sec:drrservice_ndm}). On the other hand, $b^v$ is obtained from bounds on the propagated burstiness bounds $z_c$, which is computed by TFA (i.e., $(d_c,z_c) = \tfa_c \lp\beta_c\rp $) and requires knowing $\beta_c$, the collection of DRR service curve for class $c$ at every node $v$. \cite{drr_ton} avoids this two-way dependency between TFA and DRR by restricting the analysis only to feed-forward networks.

As of today, for networks with cyclic dependencies, the only TFA solution that exists in the state-of-the-art %(other than our TFA-DRR which we present in Section~\ref{sec:initPhase}) 
is to use DRR strict service curves in the degraded operational mode (see Section~\ref{sec:drrservice_dm}) which only depend on the assigned quantum and maximum packet size of every class and hence can be computed for all classes at all nodes a priori to TFA; thus, per-class networks are independent and can be analyzed separately (i.e., the network is sliced into some per-class networks), and one instance of TFA  can be run per-class to obtain bounds; we call this method %\emph{
TFA-SOA.
%}.%; %and is the best, proven bounds in the literature; 
%see Fig.~\ref{fig:sot}.
% ; when the network is feed-forward, for each node in topological order, output burstiness bound and DRR strict service curves are computed.
% As a minor contribution, we propose an iterative method in Algorithm~\ref{alg:initPahse}, and prove its validity in Theorem~\ref{thm:tfa}, to mitigate this two-way dependency, and it serves as the initial phase of our main method.

As a first step, we propose an iterative method in Algorithm~\ref{alg:initPahse}, called TFA-DRR, and prove its validity in Theorem~\ref{thm:tfa}; it combines DRR strict service curve in non-degraded operational mode (i.e., where some arrival curves can be assumed for the interfering traffic) and TFA, and it serves as the initial phase of our main method (see Fig.~\ref{fig:sot}).

% \bleu{As of today, the only solution in the state-of-the-art (other than our TFA-DRR) is to use DRR strict service curves in the degraded operational mode (see Section~\ref{sec:drrservice_dm}) which only depend on the assigned quentua and maximum packet sizes and hence can be computed for all classes at all nodes; thus, per-class networks are independent and can be analyzed separately (i.e., the network is sliced into some per-class networks), and one instance of TFA  can be run per-class to obtain bounds; we call this method \emph{TFA-SOT} and is the best, proven bounds in the literature; see Fig.~\ref{fig:sot}.}

\subsection{Polynomial-size Linear Programming (PLP)} \label{sec:plp}

\begin{figure*}[htbp]
	\centering
	\includegraphics[width=0.85\textwidth]{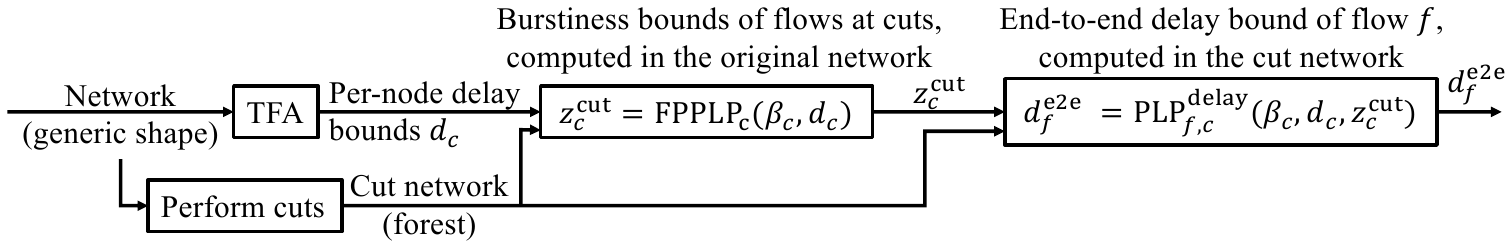}
	\caption{\sffamily \small Overview of PLP analysis for the FIFO-per-class network of some class $c$ (see items 1)-3) in Section~\ref{sec:plp}).}
	\label{fig:plp_overview}
\end{figure*} 

Polynomial-size linear program (PLP) computes end-to-end delay bounds in FIFO networks~\cite{plp}, so one instance of PLP is run per class. It requires piece-wise linear convex service curves. 
%Polynomial-size linear program (PLP) computes end-to-end delay bounds in FIFO networks with the convex service curves~\cite{plp}. In a per-class network where a convex service curve is known for every class at every node, one instance of PLP is run per class. 
PLP improves the bounds and stability region compared to TFA while remaining tractable. %It uses several linear programs.
The definition of linear programs is straightforward for tree topologies. The analysis of general topologies requires first making some cuts in the induced graph to create a forest (i.e., one or several non-connected trees). The analysis has three steps (Fig.~\ref{fig:plp_overview}): 

% \begin{enumerate}
%     \item TFA analysis to obtain per-node delay bounds $d_c$;

%     \item Then, output burstiness bounds at the cuts are computed.  %as a fixpoint of a function. In~\cite{plp}, this fixed point is computed 
% by solving one single linear program, which is equivalent to computing a fixpoint. We call $\plpcut_c(\beta_c, d_c)$ the algorithm that computes burstiness bounds of flows at cuts given strict service curves and per-node delay bounds of class $c$;

%     \item One linear program per flow of interest obtains an end-to-end delay or backlog bound; we call $\plpd_{f, c}( \beta_c, d_c,\zcutc)$ (resp. $\plpb_{f,c}(\beta_c, d_c,\zcutc) $) the algorithm that computes the end-to-end delay (resp. backlog) bound  of flow $f$ belonging to class $c$ given output burstiness bounds at the cuts, strict service curves and per-node delay bounds of class $c$.
% \end{enumerate}

1)  TFA analysis to obtain per-node delay bounds $d_c$;

2) Then, output burstiness bounds at the cuts are computed.  %as a fixpoint of a function. In~\cite{plp}, this fixed point is computed 
by solving one single linear program, which is equivalent to computing a fixpoint. We call $\plpcut_c(\beta_c, d_c)$ the algorithm that computes burstiness bounds of flows at cuts given strict service curves and per-node delay bounds of class $c$;

3) One linear program per flow of interest obtains an end-to-end delay or backlog bound; we call $\plpd_{f, c}( \beta_c, d_c,\zcutc)$ (resp. $\plpb_{f,c}(\beta_c, d_c,\zcutc) $) the algorithm that computes the end-to-end delay (resp. backlog) bound  of flow $f$ belonging to class $c$ given output burstiness bounds at the cuts, strict service curves and per-node delay bounds of class $c$.
 We use $\plpcut$ as is and use improved versions of $\plpd$ and $\plpb$, as explained in Section~\ref{sec:2imp}.
 %
 %We denote the three PLP  methods as follows: 1) $\plpcut_i(\beta_i, d_i)$ computes the valid bounds of the burstiness of flows at cuts given strict service curves and delay bounds of class $i$ for each node $v$; 2) (resp. 3))  $\mathrm{PLP}_{f, i}^{\mathrm{delay}}( \beta_i, d_i,z^{E_i^{\mathrm{cut}}})$ (resp. $\plpb_{f,i}(\beta_i, d_i,z^{E_i^{\mathrm{cut}}}) $) computes the end-to end delay (resp. backlog) bound  of flow $f$ belonging to class $i$ given output burstiness bounds at the cut, strict service curves and per-node delay bounds at each node. We will use 1) and we will improve on 2) and 3).
%Then, PLP can be applied to the cut network, which is a tree, given bounds on burstiness of flows at cuts. In this tree network, PLP computes end-to-end backlog bound and delay bound for a flow of interest. 
%Hence, we have two kinds of PLP operations: First, a PLP on the original graph induced by flows that computes bounds on the burstiness of flows at cuts. Second, given bounds on the burstiness of flows at cuts, a PLP in the tree graph included by flows to obtain end-to-end bounds.
%
As PLP uses per-node delay bounds computed by TFA, the end-to-end bounds are always better than with TFA.  
% PLP takes as inputs per-node delay bounds computed by generic TFA. This because PLP is a refinement method that given some per-node delay bonds, it can compute better end-to-end bounds (i.e., better than the sum of all per-node delay bounds in the path of flow). If per-node delay bounds are not included in PLP, the method still applies and some bounds are computed; however, there might not be as good. Also, 
 In a network with cyclic dependencies, it is possible that TFA diverges and hence the per-node delay bounds be infinite. In this case, the constraints used by PLP that involve infinite per-node delay bounds are simply always satisfied and %In spite of this, PLP generally finds delay bounds that are better than with TFA  
 %Then, adding infinite delay bounds to PLP is equivalent to not adding them at all (i.e., they become dummy constraints), and 
 PLP might or might not compute finite end-to-end bounds. In general, though, PLP finds a larger stability region than TFA, i.e., it often finds finite delay bounds when the TFA per-node delay bounds are infinite. 
 
 A detailed background on these linear programs is presented in Appendix \ref{app:plp}.

 Combining PLP and DRR strict service curves require more adaptation compared to TFA: Similar to TFA, there is a two-way dependency between DRR strict service curves and $\plpb$. Specifically, on the one hand, DRR strict service curves of node $v$ depend on 
 $b^v$,  the collection of $b_c^v$ output burstiness bound of all classes at node $v$ (see Section \ref{sec:drrservice_ndm}). On the other hand, $b^v$ is obtained using $\plpb$ which requires knowing $\beta_c$ (see item 3) in the above), the collection of DRR service curve for class $c$ at every node $v$. Thus, this creates a level of iteration between  collection  $b$ and collection $\beta$. Also, $\plpb$  uses the collection of per-node delay bounds $d$ (see item 3) in the above)  which depends on both DRR strict service curves $\beta$ and  burstiness bounds of flows at the input of nodes,  obtained using $\plpb$; this imposes another level of iteration. Moreover, $\plpb$ requires cuts and bounds on the burstiness of flows at cuts, $\zcut$ (see item 3) in the above), is obtained using $\plpcut$ where $\plpcut$ requires knowing DRR strict service curves $\beta$ and per-node delay bounds $d$ (see item 2) in the above); this imposes yet another level of iteration. Thus, we have a collection of DRR service curves $\beta$, a collection of per-node delay bounds $d$, a collection of output bound for flows at cuts $\zcut$, and a collection of burstiness bounds $b$ at the input and output of nodes, and we have different functions such as DRR strict service curves, $\plpb$, $\plpcut$, etc.  that each uses some values of these collections and improves some other values, hence, imposing different levels of iteration; it is not clear how to combine them.

 As of today, the only PLP solution that exists in the state-of-the-art %(other than our PLP-DRR which we present in Section~\ref{sec:genMethod})  
 is to use DRR strict service curves in the degraded operational mode (see Section~\ref{sec:drrservice_dm}) which only depend on the assigned quantum and maximum packet size of every class and hence can be computed for all classes at all nodes a priori to PLP; thus, per-class networks are independent and can be analyzed separately (i.e., the network is sliced into some per-class networks), and one instance of PLP can be run per-class to obtain bounds; we call this method PLP-SOA; % and is the best, proven bounds in the literature; 
 see Fig.~\ref{fig:sot}.

 We first perform the necessary adaptation of PLP to DRR by computing burstiness bounds per-class and per-output aggregate and by enabling PLP to support non-convex service curves. We then propose a generic method, PLP-DRR, in Section \ref{sec:genMethod}, for combining all these iterations sequentially and in parallel. We show, in Theorem~\ref{thm:genMethod}, that obtained bounds using our method are always valid even before convergence. Also, we show that, at convergence, the bounds are the same regardless of how iterations are combined. Lastly, we present two concrete implementations, using  a distributed computing model with shared memory (see Fig.~\ref{fig:gen_method} and Fig.~\ref{fig:versions}).

\begin{figure*}%[htbp]
	\centering
		
	\tikzset{
fnc/.style={
	rectangle,
	rounded corners,
	draw=black, very thick,
	minimum height=2em,
	text centered,
	fill=white,
	font=\scriptsize},
mem/.style={
	rectangle,
	%rounded corners,
	draw=black, very thick,
	minimum height=2em,
	text centered,
	font=\scriptsize},
ref/.style={
	%rectangle,
	%rounded corners,
	%draw=black, very thick,
	%text width=6.5em,
	minimum height=2em,
	text centered, 
	font=\bfseries\scriptsize},
	arc/.style={font=\scriptsize, sloped},
each/.style={
minimum height=2em,
	text centered,
	font=\tiny},}

\begin{tikzpicture}
	\node[fnc, text width=49.5em, text height=7.5em, dotted , ultra thin] (refine) at (6.75, 2.75) {};
	\node[fnc, text width=8em] (init) at (0, 0) {TFA analysis (Algo 1)};
	\node[ref] at (0, 0.5) {initial phase, Sec.~\ref{sec:initPhase}};
	\node [mem, text width=15em] (memory) at (6.6, 0) {{\bf shared memory} $(\beta, d, \zcut, b)$ collection of service curves, per-node delays and bursts} ;
	\node[fnc, text width=10em] (final) at (13, 0) {End-to-end delay bounds $\dend_f = \iplpd_{f,c}(\beta, d, \zcut)$};
	\node[ref] at (13, 0.6) {post-process phase, Sec.~\ref{sec:postPhase}};
	
	\node[fnc, text=white] (fpplpa) at (0.3, 2.3) {$\zcuti = \plpcut_c(\beta_c, d_c)$};	
	\node[fnc, text=white] (fpplpb) at (0.1, 2.1) {$\zcuti = \plpcut_c(\beta_c, d_c)$};
	\node[fnc] (fpplp) at (0, 2) {$\zcuti = \plpcut_c(\beta_c, d_c)$}; 
	\node[ref] at (0.3, 2.8) {bursts at cuts, Sec.~\ref{sec:plp}};
	\node[each] at (-1.6, 2.5) {\rotatebox{45}{per class $c$}};
	
	\node[fnc, text=white] (plpback) at (4.3, 3.3) {$b^g_c = \plpb_{g,c} ( \zcuti , \beta_c, d_c)$};
	\node[fnc, text=white] (plpback) at (4.1, 3.1) {$b^g_c = \plpb_{g,c} ( \beta_c, d_c,\zcuti)$};
	\node[fnc] (plpback) at (4, 3) {$b^g_c = \plpb_{g,c} ( \beta_c, d_c,\zcuti)$};
	\node[ref] at (4.3, 3.8) {aggregated bursts, Sec.~\ref{sec:plpBack}};
	\node[each] at (1.8, 3.5) {\rotatebox{45}{per class $c$,}};
	\node[each] at (2.1, 3.5) {\rotatebox{45}{per node/edge $g$}};
	
	\node[fnc, text=white] (drr) at (9.3, 3.3) {$\beta^v  = \drrb_v ( \beta^v ,b^v )$};
	\node[fnc, text=white] (drr) at (9.1, 3.1) {$\beta^v  = \drrb_v ( \beta^v ,b^v )$};
	\node[fnc] (drr) at (9, 3) {$\beta^v  = \drrb_v ( \beta^v ,b^v )$};
    \node[ref] at (9.3, 3.8) {DRR strict service curve, Sec.~\ref{sec:drrservice_ndm}};
    \node[each] at (6.8, 3.5) {\rotatebox{45}{per node $v$}};
    
    \node[fnc, text=white] (delay) at (13.3, 2.3) {$d_c^v = \pd_{v,c}(  \beta^v ,b^v)$};
    \node[fnc, text=white] (delay) at (13.1, 2.1) {$d_c^v = \pd_{v,c}(  \beta^v_c,b_c)$};
	\node[fnc] (delay) at (13, 2) {$d_c^v = \pd_{v,c}( \beta_c^v ,b_c^v)$};
	\node[ref] at (13.3, 2.8) {per-node delay, Sec.~\ref{sec:refinePhase}};
	\node[each] at (15.1, 1.9) {\rotatebox{45}{per class $c$,}};
	\node[each] at (15.2, 1.7) {\rotatebox{45}{per node $v$}};
	
	\draw[->, arc] (init) -- (memory) node[pos=0.5, above=-0.1cm] {$(\beta, d, \zcut, b)$};
	\draw[->, arc] (memory) -- (final)  node[pos=0.5, above=-0.1cm] {$(\beta, d, \zcut)$};
	
	\draw[arc, ->] ([yshift=-0.1cm] fpplp.east) -- (memory) node[pos=0.8, above = -0.1cm] {update $\zcuti$};
	\draw[arc, ->] ([xshift=-1.8cm]memory.north) -- ([xshift=1.3cm]fpplp.south) node[pos=0.5, below=-0.1cm] {read $(\beta_c, d_c)$};
	
	\draw[arc, ->] ([xshift=-0.1cm] plpback.south) -- ([xshift=-0.7cm] memory.north) node[pos=0.3, below = -0.1cm] {update $b^g$};
	\draw[arc, ->] (memory) -- (plpback) node[pos=0.5, above=-0.1cm] {read $(\beta_c, d_c, \zcuti)$};

	\draw[arc, ->] ([xshift=0.1cm] drr.south) -- ([xshift=0.7cm] memory.north) node[pos=0.2, below = -0.1cm] {update $\beta_v$};
	\draw[arc, ->] (memory) -- (drr) node[pos=0.67, above=-0.1cm] {read $(\beta^v, b^v)$};
	
	\draw[arc, ->] ([yshift=0.1cm] delay.west) -- ([xshift=1.5cm] memory.north) node[pos=0.8, above = -0.1cm] {update $d_c^v$};
	\draw[arc, ->] ([xshift=2.1cm]memory.north) -- ([yshift=-0.2cm]delay.west) node[pos=0.4, below=-0.1cm] {read $(\beta_c^v, b_c^v)$};
	
	\node[ref] at (6.5, 4.3) {refinement phase, Sec.~\ref{sec:refinePhase}};
	
\end{tikzpicture}
		\caption{\sffamily \small 
		%\jylb{Put a dotted box around the four top boxes and indicate ``refinement phase"}
		Overview of the method.
		The refinement phase consists in applying, in any order, any of the four types of refinement blocks shown on the top of the figure, which each improve the bounds stored in the shared memory. The refinement phase may be stopped using any criterion, such as convergence of the variables in the shared memory or a timeout. If stopped at convergence, the value of the shared memory is always the same, regardless of the order of the refinements.
		%We consider a shared memory  that contains $(\beta, d, \zcut, b)$. It is initialized by the bounds obtained at the initial phase in Section \ref{sec:initPhase} when using TFA. Then, at refinement phase of Section \ref{sec:refinePhase}, number of refinements run in parallel; each read some components of the shared memory and update some. Each refinement is infinite times visited, until the shared memory converges to the fixed point. Then, final values of the shared memory are post-processed in Section \ref{sec:postPhase} and end-to-end delay bounds are obtained for flows of interest. 
		}
		\label{fig:gen_method}
\end{figure*}
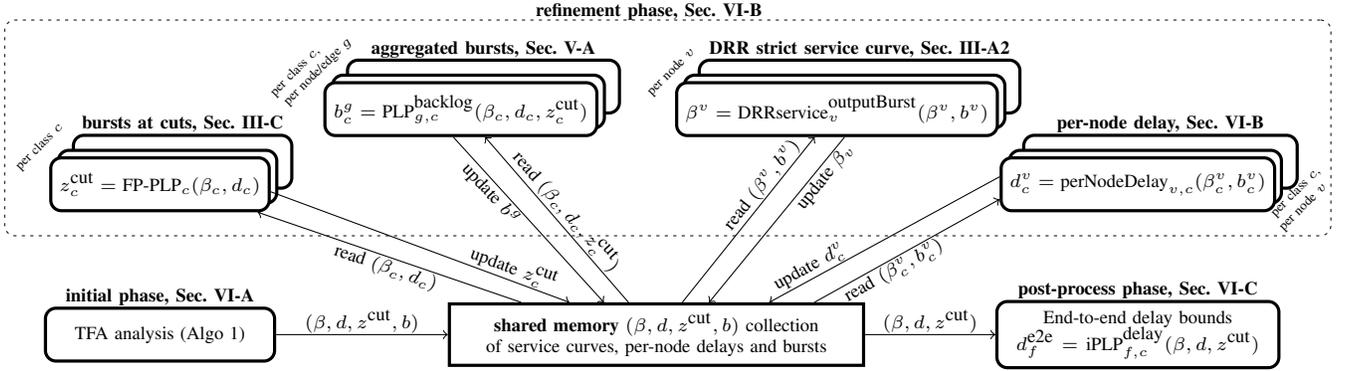

\section{Overview of the Proposed Method: PLP-DRR}
\label{sec:general_idea}
%\jylb{This now overlaps with the introduction}
Our method, called ``PLP-DRR", applies the PLP methodology to DRR and is illustrated in Fig.~\ref{fig:gen_method}. The starting point is the collection of per-class graphs and a cutset. We use the following notations:

% \begin{itemize}
%     \item $\beta = (\beta_c^v)_{1\leq c \leq n, v\in\mathcal{V}}$, is a valid collection of strict service curves of each class $c$ at each node $v$,%; all components of $\beta$ are finite.

%     \item $d = (d_c^v)_{1\leq c \leq n, v\in\mathcal{V}}$, is a valid collection of per-node delay bounds of each class $c$ at each node $v$,%; components of $d$ might be infinite.

%     \item $\zcut$, is a valid collection of output burst bounds for each flow at every edge of the cutset,%; components of $\zcut$ might be infinite.

%     \item $\zcut$, is a valid collection of output burst bounds for each flow at every edge of the cutset,%; components of $\zcut$ might be infinite.

%     \item $b = (b^g_c)_{1\leq c \leq n, g\in\mathcal{V}\cup \mathcal{E}}$ is a valid collection of burstiness bounds for aggregates of flows, indexed by a class $c$ and by some $g$. The index $g$ can be either  an edge, in which case the aggregate is the set of all flows of class $c$ carried on this edge, or a vertex $v$, in which case it is the set of all flows of class $c$ that exit the output buffer represented by $v$. %Some components of $b$ might be infinite. 
% \end{itemize}

$\sbullet[.75]$ $\beta = (\beta_c^v)_{1\leq c \leq n, v\in\mathcal{V}}$, is a valid collection of strict service curves of each class $c$ at each node $v$,%; all components of $\beta$ are finite.

$\sbullet[.75]$ $d = (d_c^v)_{1\leq c \leq n, v\in\mathcal{V}}$, is a valid collection of per-node delay bounds of each class $c$ at each node $v$,%; components of $d$ might be infinite.

$\sbullet[.75]$ $\zcut$, is a valid collection of output burst bounds for each flow at every edge of the cutset,%; components of $\zcut$ might be infinite.
    
$\sbullet[.75]$ $b = (b^g_c)_{1\leq c \leq n, g\in\mathcal{V}\cup \mathcal{E}}$ is a valid collection of burstiness bounds for aggregates of flows, indexed by a class $c$ and by some $g$. The index $g$ can be either  an edge, in which case the aggregate is the set of all flows of class $c$ carried on this edge, or a vertex $v$, in which case it is the set of all flows of class $c$ that exit the output buffer represented by $v$. %Some components of $b$ might be infinite. 

All components of $\beta$ are finite; the  components of $d$, $\zcut$ and $b$ might be infinite.

% \begin{itemize}
%     \item $\beta = (\beta_c^v)_{1\leq c \leq n, v\in\mathcal{V}}$, is a valid collection of strict service curves of each class $c$ at each node $v$; all components of $\beta$ are finite.
%     \item $d = (d_c^v)_{1\leq c \leq n, v\in\mathcal{V}}$, is a valid collection of per-node delay bounds of each class $c$ at each node $v$; components of $d$ might be infinite.
%     \item $\zcut$, is a valid collection of output burst bounds for each flow at every edge of the cutset; components of $\zcut$ might be infinite.
%     \item $b = (b^g_c)_{1\leq c \leq n, g\in\mathcal{V}\cup \mathcal{E}}$ is a valid collection of burstiness bounds for aggregates of flows, indexed by a class $c$ and by some $g$. The index $g$ can be either  an edge, in which case the aggregate is the set of all flows of class $c$ carried on this edge, or a vertex $v$, in which case it is the set of all flows of class $c$ that exit the output buffer represented by $v$. Some components of $b$ might be infinite. 
% \end{itemize}

As PLP requires per-node delay bounds, the method starts with an initial phase that performs a TFA analysis of the original (uncut) network. Note that DRR requires the service curve collection $\beta$ to be computed from output burstiness bounds, which we derive from the TFA analysis; hence, we apply an iterative procedure, which we prove to be valid. At the end of this initial phase, we have a collection of service curves $\beta$ and per-node delay bounds $d$, from which some propagated burstiness bounds (i.e., burstiness bounds for all flows at every output), hence $\zcut$ and $b$, can be derived. 

The next phase in the classical PLP methodology, FP-PLP, computes a fixpoint $\zcut$ by solving a linear program. Here, however, a new value of $\zcut$ allows to compute better output burstiness bounds in $b$ using another linear program with $\plpb$, which, in turn, allows to compute better service curves $\beta$ using the DRR service curve method recalled in Section~\ref{sec:drrservice}. The linear program in $\plpb$ uses per-node delay bounds $d$, which can in turn be improved by re-running TFA whenever $\beta$ is improved, and then $\plpb$ could also be re-run. Also, better $\beta$ and $d$ allows FP-PLP to re-compute a better fixpoint $\zcut$, which can in turn be used to improve all other variables. Therefore, the second phase of the PLP methodology needs to apply a number of refinements again and again. Instead of proposing a specific arrangement of the refinements, we propose to perform them in any arbitrary order, using a shared-memory model (Fig.~\ref{fig:gen_method}).
% , as illustrated in Fig.~\ref{fig:gen_method}. 
As we show in Section~\ref{sec:refinePhase}, all refinements provide valid bounds, therefore the method can be stopped at any time. However, we show that it converges to bounds (some of them possibly infinite) that are independent of the arrangement of the refinements.

The ``DRR strict service curve" block in the refinement phase uses as input some burstiness bounds for the aggregate of all flows of all interfering classes at the output of a node. Such a bound could be obtained by using the existing version of $\plpb$ applied to all flows in the aggregate. We improve both the obtained bound and the computing time by using the modification of $\plpb$ described in Section~\ref{sec:plpBack}.

The third phase of PLP is to obtain end-to-end bounds by applying one instance of $\plpd$ to every flow of interest. Here, we use $\plpd$ with one improvement (iPLP), which allows to use non-convex service curves at the expense of adding a few binary variables to the linear program (Section~\ref{sec:iplp}). Such an improvement could also be used in the refinement phase, but we found experimentally that this would have no noticeable effect.

\section{Two improvements to PLP}
\label{sec:2imp}
In this section, we first show how to compute an upper bound of aggregate burstiness of flows and how to include some non-convex service curves in the PLP analysis. Note that our two improvements concern $\plpb_{f,i}$ and $\plpd_{f,i}$ in step 3) of PLP as explained in Section \ref{sec:plp}; specifically, in this section, we assume that we have a collection of trees where  TFA per-node delay bounds and bounds on the burstiness of flows at cuts are already obtained.

Let us first briefly present the linear programs used by PLP, more details can be found Appendix~\ref{app:plp}. PLP considers the arrival and departure time of a bit of interest; it then derives a number of time instants at every node, each of which is represented by a variable in PLP. To every time instant at a node is also associated a variable that represents the value of the cumulative arrival function of the flows at this time instant. Network calculus relations such as arrival curve constraints, FIFO constraints and service curve constraints are translated into linear constraints; 
the objective function to be maximized is the delay or backlog of the flow of interest.

Recall that at this step we use the cut network and thus assume that the graph induced by flows is a collection of non-connected trees, and the analysis is done on every tree (with edges directed towards the root).
%otherwise some artificial cuts are made and $\plpcut$ is used to obtain bounds on the burstiness of flows at cuts; each tree is a directed graph with edges towards the root. 
It follows that each node $v$, except the root, has a unique successor. We let $\suc(v)$ denote the successor of node $v$, and add an artificial node $v_0$ to be the successor of the root. Define the depth of nodes $\depth$ as follows: $\depth(v_0) = 0$
  %$\depth(\pi_{\bar{f}}(l_{\bar{f}})+1) = 1$ 
  and for every $v$, $\depth(v) = \depth\lp \suc\lp v \rp  \rp + 1$. 
  
 {\em  Time Variables:} For every node $v$, define  $\tvar{v}{k}$ with $k \in \{0, \ldots, \depth(v) \}$. 
 %Map to $\mathbf{t}_{(j,k)}$ of \cite{plp} by mapping $j$ to $v$.

 {\em Process Variables:} For every node $v$ and $v'$ and every $f$ at $v$, define $\fvar{v}{f}{v'}{k}$ with $k \in \{0, \ldots, \depth(v') \}$, where  $\fvar{v}{f}{v'}{k}$ is a variable for the cumulative arrival function of flow $f$ at the input of node $v$ at time $\tvar{v'}{k}$. 
 %It can be mapped to $\mathbf{F}^j_i\mathbf{t}_{(j',k)}$ of  \cite{plp} by mapping $j$ to $v$, $i$ to $f$, and $j'$ to $v'$.
%We map our notations to the ones used in \cite{plp} to enable an interested reader to completely implement our methods. 
In the next paragraphs, we only presents the parts of the linear program that are modified. The complete linear programs are presented in Appendix \ref{app:plp}.
%\jylb{Bold face usually represents a vector}

\subsection{PLP to Upper-bound the  Aggregate Burstiness of Flows} 
\label{sec:plpBack}
We show that the same PLP, used to compute a backlog bound for a single flow, can be used to compute a backlog bound for the aggregate with some modifications: Consider a set of flows of interest $F$, whose destination is the root of the tree.
   Modify the PLP used to compute a backlog bound for a single flow as follows: 
%  We use same PLP used to compute a backlog bound for a single flow, and modify it as follows: 
  Let $v_f$ be the first node visited by flow $f$ in the tree for all $f\in F$. 
  \begin{itemize}
    \item  Additional constraints: $\forall f\in F$, $\forall k,~ k \in [0,~\depth(v_f)],~ \fvar{v_f}{f}{v_0}{0} - \fvar{v_f}{f}{v_f}{k} \leq b_f + r_f(\tvar{v_0}{0} - \tvar{v_f}{k})$;
    \item New objective: Maximize $\sum_{f\in F} \lp\fvar{v_f}{f}{v_0}{0} - \fvar{v_0}{f}{v_0}{0} \rp $. 
  \end{itemize}
   %We use the same PLP used to compute a backlog bound for a single flow yet the objective is now for the aggregate, i.e., a backlog bound for aggregate of flows $F$ in the network. With $v_f$ being the first node visited by flow $f$ in the tree,  
   %Consider an edge $e = (v,u)$; $v$ is the root of the tree.  Then, 
   %we add the following constraints for each $f\in F$: %$\forall k \in \{0, \ldots, \depth(\pi_f) \}, \fvar{\pi_f(1)}{f}{v+1}{0} - \fvar{\pi_f(1)}{f}{\pi_f(1)}{k} \leq b_f + r_f(\tvar{v+1}{0} - \tvar{\pi_f(1)}{k})$. 
   %$\forall k \in \{0, \ldots, \depth(v_f) \}, \fvar{v_f}{f}{v_0}{0} - \fvar{v_f}{f}{v_f}{k} \leq b_f + r_f(\tvar{v_0}{0} - \tvar{v_f}{k})$. 
   
   %Then the objective is to maximize $\sum_{f\in F} \lp\fvar{v_f}{f}{v_0}{0} - \fvar{v_0}{f}{v_0}{0} \rp $. 
   %$\sum_f \lp\fvar{\pi_f(1)}{f}{v+1}{0} - \fvar{v+1}{f}{v+1}{0} \rp $. 
 
%\jylb{To which tree do we apply this ? What is the root ?}    
We call $\plpb_{v,c}$ the resulting program, when applied to class $c$ and to the sub-tree of the cut network rooted at some node $v$. Here $F$ is the set of flows that exit node $v$ and the program obtains the aggregate burstiness of flows of class $c$ at the output of node $v$; the results are used by $\drrb_v$ in the refinement phase to compute DRR strict service curves. 
We also apply this program when $e$ is an edge, and also call it $\plpb_{e, c}$. Here, $F$ is the set of flows that use edge $e$ and the sub-tree is rooted at the node that edge $e$ exits. The results are used in the refinement phase by $\pd_{v,c}$ %\jylb{$D$ or perNode Delay} 
to compute per-node delay bounds.
% Theorem~\ref{thm:plpBack} is proved in Section \ref{sec:proof-plpBack}.
%   \anne{Question: backlog at the output of a node or at the edge?}
   
   \begin{theorem}[PLP to Upper-bound the  Aggregate Burstiness of Flows]\label{thm:plpBack}
    The solution of $\plpb_{g,c}$ is a valid bound on aggregate burstiness of flows carried by edge or node $g$.
       %Consider $\plpb_{g,i}$ the PLP constructed above and let $\mathcal{B}$ be what it returns. Then, $\mathcal{B}$ is a valid bound on aggregate  burstiness of flows carried by edge $e$.
   \end{theorem}
   
  The proof is in Section \ref{sec:proof-plpBack}, and the key idea of the proof is as follows: We show that a valid backlog bound for an aggregate of some flows is also a valid burstiness bound for the aggregate. This is obtained in Lemma \ref{lem:aggBacklog} where we show that in any acceptable trajectory scenario of the system (i.e., a set of valid input/output processes for all flows), the burstiness of the aggregate never exceeds the backlog bound.%Then, observe that the new objective we use is a backlog bound for the aggregate flow hence the solution is a valid burstiness bound for the aggregate flow.  

   \subsection{iPLP: a PLP that supports non-convex service curves}\label{sec:iplp}
   Our goal here is to modify PLP, used to compute a delay, such that it can handle a non-convex service curve expressed, at node $v$ and class $c$,  as $\max(\beta_c^v, \beta^{\scriptsize\textrm{nc},v})$, where $\beta_c^v$ is  piece-wise linear convex (i.e., $\beta_c^v = \max_p \beta_{R_p^v, T_p^v})$),
   %that can be written as $\max_p \beta_{R_p^v, T_p^v}$, 
   and $\beta_c^{\scriptsize\textrm{nc}, v} = \min(\beta_{R^v, T_1}, q_c^v)$ as described in Section~\ref{sec:drrservice}. 
    Similar to the previous case, we present only the parts of the linear program that are modified, namely the service constraints.
    
  %\textbf{iPLP Construction:}  Recall the notation of Section \ref{sec:plpB}. To construct iPLP, we use the exact same PLP for delay bounds but we add some new constraints to the service curve constraints. In this paper, we only present the service curve constraints of PLP, and the details can be found in Section 4 of \cite{plp}. We map our notation to the ones used in \cite{plp} to enable an interested reader to completely implement iPLP. 
  
  For the sake of concision, we now introduce the variables and constraints $\mathbf{At}_u^v = \sum_{f \in \textrm{In}(v)} \fvar{u}{f}{u}{\scriptsize\depth(u)}$ and  $\mathbf{At}_v^v =\sum_{f \in \textrm{In}(v)} \fvar{v}{f}{v}{\scriptsize\depth(v)}$, where $u = \suc(v)$.
   
   %Consider a node $v$ and class $C_i$. As $\beta^v_i$ is piece-wise linear and convex, it can be written as maximum of number of rate-latency function. Assume $\beta^v_i = \max_p\beta_{R^v_p,T^v_p}$. 
   
   The original service curve constraints of PLP are kept: 
   %can be written as follows: Let $u = \suc(v)$ then
   \begin{align}
        \mathbf{At}_u^v - \mathbf{At}_v^v 
        %\sum_{f \in \textrm{In}(v)}\lp \fvar{u}{f}{u}{\scriptsize\depth(u)} - \fvar{v}{f}{v}{\scriptsize\depth(v)} \rp
        &\geq 0;\label{eq:sc1}\\
        \forall p,~\mathbf{At}_u^v - \mathbf{At}_v^v 
        %\sum_{f \in \textrm{In}(v)}\lp \fvar{u}{f}{u}{\scriptsize\depth(u)} - \fvar{v}{f}{v}{\scriptsize\depth(v)} \rp
        &\geq R^v_p\lp \tvar{u}{\scriptsize\depth(u)} - \tvar{v}{\scriptsize\depth(v)}  - T^v_p  \rp. \label{eq:sc2}
   \end{align}
 % For sake of conciseness, we now introduce the variables and constraints $\mathbf{At}_u = \sum_{f \in \textrm{In}(v)} \fvar{u}{f}{u}{\scriptsize\depth(u)}$ and  $\mathbf{At}_v =\sum_{f \in \textrm{In}(v)} \fvar{v}{f}{v}{\scriptsize\depth(v)}$.

   As $\beta^{\scriptsize\textrm{nc},v}$ is the minimum of a rate-latency function and a constant (see Fig.~\ref{fig:ncService}); the implementation of $\beta^{\scriptsize\textrm{nc},v}$   requires a ``if then else'' structure: If $\tvar{u}{\scriptsize\depth(u)} - \tvar{v}{\scriptsize\depth(v)} \leq T_2$, then we must have $\mathbf{At}_u^v - \mathbf{At}_v^v \geq R^v(\tvar{u}{\scriptsize\depth(u)} - \tvar{v}{\scriptsize\depth(v)}-T_1)$. Else, if $\tvar{u}{\scriptsize\depth(u)} - \tvar{v}{\scriptsize\depth(v)} > T_2$,  we must have $\mathbf{At}_u^v - \mathbf{At}_v^v \geq q_c^v $. This can be  modeled by 
   means of a binary variable in a linear program.  Define $\textbf{b}_v \in \{0,1\}$ and consider a large enough positive $M$, and add the following constraints:
   % Define $\textbf{b}_v \in \{0,1\}$ and consider a large enough positive $M$, and add the following constraints:  
   %Then, we add the followings that accounts for $\beta^{\scriptsize\textrm{nc},v}$ as shown in Fig.~\ref{fig:ncService}.: 
\begin{align}
     \mathbf{At}_u^v - \mathbf{At}_v^v &\geq R^v(\tvar{u}{\scriptsize\depth(u)} - \tvar{v}{\scriptsize\depth(v)}-T_1) - M\textbf{b}_v; \label{eq:nc1}\\
     \mathbf{At}_u^v - \mathbf{At}_v^v &\leq q_c^v + M\textbf{b}_v; \label{eq:nc3}
     \\
     \mathbf{At}_u^v - \mathbf{At}_v^v &\geq q_c^v - M(1-\textbf{b}_v ).\label{eq:nc2}
\end{align}

% \bleu{As $\beta^{nc, v}_c$ is the minimum of a rate-latency function and a constant (see Fig.~\ref{fig:ncService}); the implementation of $\beta^{nc, v}_c$   requires a ``if then else'' structure: This can be modeled by means of a binary variable $\textbf{b}_v$ in a linear program.}

We let  $\iplpd_{f,c}$ denote this Integer, Polynomial-sized Linear Program. Theorem~\ref{thm:iplp} is proved in Section \ref{sec:proof-iplp}.
\begin{theorem}[iPLP: a PLP that supports non-convex service curves] \label{thm:iplp}
    Consider iPLP as constructed above.  Then, 1) iPLP gives a valid delay bound for the flow of interest and 2) the bound is less than or equal to that of PLP.
\end{theorem}

Note that iPLP solves a Mixed-Integer Linear-Program (MILP), and the MILP solver we use does not guarantee that it finds the optimal solution. However, it guarantees that the solution is feasible, and it indicates whether the obtained solution is optimal. In all examples we tested, we always obtained the optimal solution (see Fig.~\ref{fig:numFull}~(d)).

%The proof is in Section \ref{sec:proof-iplp}.

\section{Our Proposed Method: PLP-DRR}
\label{sec:detail}
In this section, we provide the details of our generic method and we present two concrete implementations.

\subsection{Initial phase: TFA-DRR}
\label{sec:initPhase}
 
The goal of the first phase is to provide valid values for $(\beta, d,\zcut, b)$, using TFA. Specifically, we want to analyze the original (uncut) network using TFA (note that TFA itself does not necessarily require cuts~\cite{sync_TFA}). 
In networks with cyclic dependencies, the TFA analysis of a DRR system of~\cite{drr_ton} cannot be directly applied here, as propagated burstiness bounds are needed to compute the DRR strict service curves and vice-versa.
However,  it is possible to first compute DRR strict service curves without assumption of the arrival curves using $\betacdma_c$ for each node $v$ and class $c$: these service curves only depend on the fixed parameters such as assigned quanta, the aggregate strict service curve, and maximum packet sizes of classes at a node, as explained in Section~\ref{sec:drrservice_dm}. From there, one can iterate between the computation of output bursts (used to deduce the arrival curves) and the DRR strict service curves that take into account the arrival curves. 
The method is described in Algorithm~\ref{alg:initPahse}: The local variable $z$ represents the propagated burstiness of all flows at all outputs. First, at line~\ref{line:betaCDM}, initial strict service curves of DRR in degraded operational mode are computed at all nodes for all classes. Then, the algorithm alternates %corresponds to alternating 
between performing a TFA analysis and computing new DRR strict service curves. 
More precisely, at line~\ref{line:tfa}, a TFA analysis (explained in Section \ref{sec:tfa}) is performed for each class, with the previously computed service curves; hence, some bounds on propagated burstiness of flows $z$ and per-node delay bounds $d$ are computed and
%. These burstiness bounds can be 
used to compute arrival curves at each node for each class (line~\ref{line:arrival}). Then, at line~\ref{line:drrService}, DRR strict service curves are improved by taking into account these arrival curves. This procedure continues until we reach a stopping criteria; for example, when each component of vector $d$ decreases insignificantly. Note that computed  bounds are valid at each iteration. % and we do not need to necessarily wait for convergence, and thus any stopping criteria is valid.
At this point, TFA analysis is completed, however, we need to compute bounds on the aggregate burstiness of flows of every class either at each edge (including at the cutset) and at the output of every node, as this is used in the refinement phase. This is performed at lines 7-11. 
The delay and burstiness bounds computed by the TFA analysis at line~\ref{line:tfa} might not be finite. For example, after the first execution of line~\ref{line:tfa}, some classes might provide finite delay bounds (called stable classes) and other infinite delay bounds (called unstable classes). At the first execution of lines~4-6, the DRR strict service curves of the unstable classes are improved using the arrival curves of the stable classes. Then, at  next iteration, more stable classes might be obtained and so on. 
	\begin{algorithm}[tbp]
	
    \SetKwInOut{IV}{Local Variable}
	%\SetKwInOut{Output}{Output}
	\KwResult{Initial values of $\lp \beta, d, \zcut,b\rp$}
	\IV{$z$, collection of bounds on the propagated burstiness of flows}
% 	 \KwResult{$\lp \beta, d, \zcut,b\rp$: Collection of per-node strict service curve, per-node delay bounds,  bounds on propagated burstiness of flows at cuts, and  aggregate burst bound at the output of every node and every edge for all classes}
    \lFor{node $v$ $\leftarrow$ $1$ to $|\mathcal{V}|$}{
    	 $\beta^v  \gets \betacdma$ \label{line:betaCDM}
    }
    % \BlankLi
	\While{Stopping criteria not reached}{
		\lFor{each class $c$}{
	     $\lp d_c ,z_c\rp \gets\tfa_c \lp\beta_c\rp$ \label{line:tfa}
	    }
% 	\BlankLine
	    \For{node $v$ $\leftarrow$ $1$ to $|\mathcal{V}|$}
	    {
	    \lFor{each class $c$}{
	        compute $\alpha^{v}_c$ from $z^{\textrm{In}_c(v)}_c$ \label{line:arrival}
	       % $b_i^v \gets \sum_{e\in\textrm{In}(v)} \lp \sum z_i^e \rp $
	       % $\alpha^{v}_i \gets \sum_{e\in\textrm{In}(v)} \gamma_{r_f, b_i^v}$\;
	        }
    	$\beta^v \gets \drri_v \lp \alpha^{v}\rp$\;\label{line:drrService}
        }
    \label{line:iter}
	}
% 	\BlankLine
	\lFor{each class $c$}{
	$\zcutc \gets z^{E^{\scriptsize\textrm{cut}}_c}_c$
	}
	
	\For{node $v$ $\leftarrow$ $1$ to $|\mathcal{V}|$}{
	    \For{each class $c$ and each $e\in \textrm{In}_c(v)$}{
	        	    $b_c^e \gets \sum z^e_c$\;
	        }
	        $b_c^v \gets \sum z^{\textrm{Out}_c(v)}_c$ \; 
	}
	
    \textbf{return} $(\beta, d, \zcut,b)$
    \caption{Initial Phase: TFA-DRR}
	\label{alg:initPahse}
\end{algorithm}
%\jylb{$z$ is mysterious. It is an internla variable and must be initialized.} 
%\hossein{fixed}
\begin{theorem}[Correctness and Convergence of TFA-DRR] \label{thm:tfa}
Consider a network with DRR scheduling per class, as described in Section \ref{sec:sysmodel}, and consider Algorithm~\ref{alg:initPahse}. Then,  1) $\lp \beta, d, z\rp$ (and the resulting $\zcut \mand b $ obtained from $z$ at lines 7-11)  are valid bounds at every iteration at lines 2-6, and 2) they converge as the number of iterations goes to infinity. Note that some values of $d, z$ (and the resulting $\zcut \mand b $) might be infinite.
%  \begin{enumerate}
%     % \item $(\beta, d, z)$ converges to a fixed-point with possibly infinite values for some components $d$ and $z$.
%     % $\beta$ at each iteration is valid and finite, and it converges to a fixed-point. is guaranteed to be finite and is valid strict service curves.$(\beta, d, z)$ converges to a fixed-point with possibly infinite values for some components $d$ and $z$.
%     \item $(\beta, d, z)$ are valid bounds at every iteration.
%     \item They converge.
%     \item Some values of $d, z$ may be infinite.
%     % \item   $\beta$ at each iteration is valid and finite, and it converges to a fixed-point.
%     % \item  Values of  $d$ and $z$ might or might not be finite; Some values of  $d$ and $z$  converges to a finite  fixed-point and are valid bounds, and some diverge to infinite values.
%     % \anne{In an infinite loop lines (8-13)??? Some values of  $d$ and $z$  converges to a finite value and are valid bounds, and some remain infinite.}
    
%     % \jylb{It seems that we can simplify the statement: (1) the bounds are valid at every iteration (2) they converge (3) some values of $d, z$ may be infinite. }
%  \end{enumerate}
\end{theorem}

The proof is in Section \ref{sec:proof-shared-memory}.
% The proof of Theorem \ref{thm:tfa} is straightforward: The correctness is implied by the correctness of DRR strict service curve and TFA analysis; also, it is guaranteed to converge, as for each per-node delay bound, a non-increasing sequence is obtained; as delay bounds are positive and the number of per-node delay bounds are finite it converges. 
% Note that computed  bounds are valid at each iteration (if finite), and we do not need to necessarily wait for convergence. 
At this point we have obtained a value of $(\beta, d, \zcut, b)$ that constitute valid bounds. %The next phase of the method is to improve $(d, \beta, \zcut, b)$ using PLP

\subsection{Refinement phase: PLP and parallelization}
\label{sec:refinePhase}
The next phase of the method is to improve the value of $(d, \beta, \zcut, b)$ using the PLP methodology. 
As mentioned in Section~\ref{sec:general_idea}, the variables $(\beta, d, \zcut , b)$ are interdependent, and can be improved by some refinements that we list here: 

% \begin{itemize}
%     \item $\zcuti \gets \plpcut_c(\beta_c, d_c)$: computes burstiness bounds at cuts for class $c$ (Section~\ref{sec:plp});  

%     \item $b^g_c \gets \plpb_{g,c}(\beta_c,d_c,\zcuti)$: computes burstiness bounds of the aggregate flows of class $c$ at the output of node $g$, if $g\in \mathcal{V}$, or carried by  edge $g$, if $g\in \mathcal{E}_c$ (Section~\ref{sec:plpBack});

%     \item ${\beta}^v \gets \drrb_v \lp \beta^v ,b^v  \rp$: computes the DRR strict service curves at node $v$ given the output burstiness bounds at this node (Section~\ref{sec:drrservice}).

%     \item $d^v_c \gets \pd_{v,c}(\beta^v_c,b_c)$:
% \end{itemize}

$\sbullet[.75]$ $\zcuti \gets \plpcut_c(\beta_c, d_c)$: computes burstiness bounds at cuts for class $c$ (Section~\ref{sec:plp});  

$\sbullet[.75]$ $b^g_c \gets \plpb_{g,c}(\beta_c,d_c,\zcuti)$: computes burstiness bounds of the aggregate flows of class $c$ at the output of node $g$, if $g\in \mathcal{V}$, or carried by  edge $g$, if $g\in \mathcal{E}_c$ (Section~\ref{sec:plpBack});

$\sbullet[.75]$ ${\beta}^v \gets \drrb_v \lp \beta^v ,b^v  \rp$: computes the DRR strict service curves at node $v$ given the output burstiness bounds at this node (Section~\ref{sec:drrservice}).

$\sbullet[.75]$ $d^v_c \gets \pd_{v,c}(\beta^v_c,b_c)$:
% $d^v_c \gets \mathcal{D}_{v,c}(\beta^v_c,b_c)$ $d^v_c \gets \pd_{v,c}(\beta^v_c,b_c)$: 
%\jylb{Why do we call it $\mathcal{D}$ and not something like ``perNodeDelay" as would be consistent with other names of blocks. If we insist on using $\mathcal{D}$, it should be in the notation list. }\hossein{I fix it everywhere} 
computes the per-node delay of node $v$ for class $c$. This is the horizontal distance between the arrival curve and the service curve. For each input edge $e$ of node $v$, in addition to the aggregate burstiness $b_c^e$,  a rate limitation imposed by the link can be used to improve the arrival curve. This is known as line-shaping  %\cite{Ahlem-line-shaping,Grieu-line-shaping, bouillard2020tradeoff}.
    \cite{plp,tfa++,Grieu-line-shaping}.
    Then, we take into account the effect of line shaping and of the packetizer as in Section VI-B of \cite{ludo_cycle}.

All these refinements can be applied in any order, and it is not clear in what order we should use. We avoid the issue by first presenting a generic scheme that uses a distributed computing model with a  shared memory, and we show that the values converge to the same values regardless of the order of the operations under mild assumptions. We also then present two practical, concrete implementations of this scheme with parallel-for loops. 

% \subsubsection{Generic Scheme Using Shared Memory Parallelization}
\subsubsection{Generic Scheme: A Distributed Computing Model with Shared Memory}
 \label{sec:genMethod}

The generic scheme is presented in Fig.~\ref{fig:gen_method}. We consider a distributed system with a shared memory and a finite %of $(\beta, d, \zcut, b)$
number of workers (processes or threads). The shared memory stores the current value of $(\beta, d, \zcut, b)$; it is initialized with the result of the initial phase described in Section \ref{sec:initPhase}. Every worker has read and write access to the shared memory, and whenever a worker is free and decides to work, it performs the following steps: 

% \begin{itemize}
%     \item It chooses a refinement in the list above, let us call it $h$; for example, it may choose $\plpcut_c()$ for some class $c$, or $\plpb_{g,c}()$ for some $c,g$, etc. 

%     \item The worker then makes a read-only operation on the shared memory in order to obtain the value, say $x$, of its argument. For example, if $h$ is $\plpcut_c()$, then the worker reads $x=(\beta_c, d_c)$. We assume that read-only operations are atomic, i.e., the values read by the worker cannot be modified by other workers during the read operation (such that the worker has a valid snapshot).

%     \item The worker computes $y=h(x)$, i.e., a new value of some of the bounds in the shared memory. For example, if $h$ is $\plpcut_c()$, the worker computes $y=\zcuti$.

%     \item The worker asks for a read/write lock on the shared memory. Such a lock prevents other workers from writing into the memory until the lock is released by this worker. Once the lock is obtained, the worker reads the current value $y'$ of the same variables it wants to update from the shared memory, computes the component-wise minimum of $y$ and $y'$, writes the resulting values into the shared memory, and releases the lock. The minimum is computed because some other worker might have improved the same value during the computing time of this worker.

%     \item The worker is now free and might decide to work again.
% \end{itemize}

$\sbullet[.75]$ It chooses a refinement in the list above, let us call it $h$; for example, it may choose $\plpcut_c()$ for some class $c$, or $\plpb_{g,c}()$ for some $c,g$, etc. 

$\sbullet[.75]$ The worker then makes a read-only operation on the shared memory in order to obtain the value, say $x$, of its argument. For example, if $h$ is $\plpcut_c()$, then the worker reads $x=(\beta_c, d_c)$. We assume that read-only operations are atomic, i.e., the values read by the worker cannot be modified by other workers during the read operation (such that the worker has a valid snapshot).

$\sbullet[.75]$ The worker computes $y=h(x)$, i.e., a new value of some of the bounds in the shared memory. For example, if $h$ is $\plpcut_c()$, the worker computes $y=\zcuti$.

$\sbullet[.75]$ The worker asks for a read/write lock on the shared memory. Such a lock prevents other workers from writing into the memory until the lock is released by this worker. Once the lock is obtained, the worker reads the current value $y'$ of the same variables it wants to update from the shared memory, computes the componentwise minimum of $y$ and $y'$, writes the resulting values into the shared memory, and releases the lock. The minimum is computed because some other worker might have improved the same value during the computing time of this worker.

$\sbullet[.75]$ The worker is now free and might decide to work again.

Note that some delay bounds and burstiness bounds might be infinite. For example, initial bounds obtained by TFA might be infinite for a class (but they might become finite after the operations of some workers). 
 
This generic scheme does not prescribe any specific arrangement of how the workers are scheduled. But, for convergence, we assume \textbf{hypothesis (H)}: In a hypothetical execution of infinite duration, for every time $t>0$ and every refinement $h$, there exists a time $s>t$ at which one worker starts working and chooses $h$.

\begin{theorem}[Correctness and Convergence of PLP Refinement Phase of Section \ref{sec:genMethod}] \label{thm:genMethod}
 Consider a network with DRR scheduling per class, as described in Section \ref{sec:sysmodel}, and  consider the generic method described above. 
%  Assume $(\beta,d,\zcut, b)$ are valid per-node and  per-class DRR strict service curves,  per-node and per-class delay bounds, and per-class propagated burstiness at cuts, edges and nodes; some values of $d,\zcut,\mand b$ might possibly be infinite.  Consider the generic method described above where the shared memory is initialized  at time $0$ with \jylb{tautology} $(\beta, d, \zcut, b)$. 
 Let $(\beta^t,d^t,z^{\scriptsize\textrm{cut}, t}, b^t)$ be the value of the shared memory at time $t>0$. Then,
  \begin{enumerate}
    \item $(\beta^t,d^t,z^{\scriptsize\textrm{cut}, t}, b^t)$ are valid bounds; some values of $( d^t,{\zcut}^t, b^t) $ might be infinite.
    \item The limit of $(\beta^t,d^t,z^{\scriptsize\textrm{cut}, t}, b^t)$  as $t\to\infty$ exists (call it  $(\beta^*,d^*,z^{\scriptsize\textrm{cut}, *}, b^*)$). Some components of $d^*,z^{\scriptsize\textrm{cut}, *}, \mand b^*$ might be infinite. 
    \item Given the initial value of the shared memory, the limit $(\beta^*,d^*,z^{\scriptsize\textrm{cut}, *}, b^*)$ is independent of the order and the execution time of every refinement.
 \end{enumerate}

%  \begin{enumerate}
%     \item $\beta^t$ is valid and finite; also, finite components of $(d^t,{\zcut}^t, b^t)$ are valid bounds \jylb{??? Are infinite components invalid bounds ?}.
%     \item The limit of $(\beta^t,d^t,{\zcut}^t, b^t)$  as $t\to\infty$ exists (call it  $(\beta^*,d^*,{\zcut}^*, b^*)$). Some components of $d^*,{\zcut}^*, \mand b^*$ may be infinite. 
%     \item Given the initial value of the shared memory, the limit $(\beta^*,d^*,{\zcut}^*, b^*)$ is independent of the order and the execution time of every refinement.
%  \end{enumerate}
\end{theorem}

Theorem~\ref{thm:genMethod} assumes that each refinement is performed infinitely often (hypothesis \textbf{(H)}). Otherwise, obtained bounds in the limit will be larger than or equal to those obtained by our scheme.

%We now present two implementations of the method, where we explain in what order we can apply the refinements. By Theorem \ref{thm:genMethod}, we know the order does not effect the final result, however, in our implementations we 
\subsubsection{Two Implementations of PLP Refinements}

We presented the generic presentation of this phase. By Theorem \ref{thm:genMethod}, any implementation results in the same final bounds. In this section, we present two concrete implementations.  

Both implementations have two main blocks: computing output burstiness bounds at cuts (with $\plpcut$), and locally improving the per-node delays and DRR strict service curves. The main difference between the two implementations is when to switch between these two blocks of operations.

For each class $c$, after removing edges $\ecutc$ in $\mathcal{G}_c$, we obtain a collection of trees. Let $\mathcal{T}$ be the collection of trees of all classes. The second block, called $\textrm{DRR}^{\textrm{tree}}_{T}$, is executed in parallel on each tree $T\in\mathcal{T}$: it is described in Algorithm~\ref{algo:tree-par}.

% We assume that the cuts are the same for each graph $\mathcal{G}_c$.

% \jylb{Inconsistent. In Section II and Fig2. we spend a lot of effort explaining that the cuts need not be the same for every graph and in all our implementations we impose that they are the same.}
% \hossein{We do not need the cuts to be the same: For each class $c$ we have a set of trees; we collect all trees of all classes. Some nodes appears in several trees but does not cause a problem; they are visited more often.}
% \anne{Having different trees does not cause problems, but the reason of this choice is to minimize the calls an do it in order, to obtain better improvement than in a random order. The decomposition of Fig. 2 follows this scheme.}\jylb{OK, so we may do different cuts but we require all trees to be same. The generic method seems to be valid in all cases, perhaps we should leave this discussion to the journal version}
% In other words, a cut is first computed on the graph induced by all flows ($\mathcal{G} = (\mathcal{V}, \cup_c \mathcal{E}_c)$), and we call $\mathcal{T}$ the set of trees obtained after removing the edges $E^{cut}$. The second block is executed in parallel on each tree $T\in\mathcal{T}$, is called $\mathrm{DRR}^{\mathrm{tree}}_{T}$, and described in Algorithm~\ref{algo:tree-par}.

The algorithm is based on the observation that in a feed-forward topology, when service curves and per-node delay bounds are computed in the topological order, there is no need for iterations. 
Some operations are performed in the topological order of the nodes of the tree (the operation on one node must wait that the operation on its predecessors are finished). 
For each node $v$, there are two steps in sequence. The former, at lines 2-4, computes the DRR strict service curve for all classes. This operation requires the refinement of the output burstiness bounds $b^v_c$ for each class, and can be computed in parallel (lines 2-3). The latter, at lines 5-8,  improves the per-node delays of node $v$ based on the newly computed service curves. Again, this operation requires the refinements of the burstiness bounds at all input edges of the node, that can be computed in parallel (lines 6-7).
% \algblock{ParFor}{EndParFor}
% % customising the new block
% \algnewcommand\algorithmicparfor{\textbf{parfor}}
% \algnewcommand\algorithmicpardo{\textbf{do}}
% \algnewcommand\algorithmicendparfor{\textbf{end\ parfor}}
% \algrenewtext{ParFor}[1]{\algorithmicparfor\ #1\ \algorithmicpardo}
% \algrenewtext{EndParFor}{\algorithmicendparfor}
\begin{algorithm}[tbp]
    \KwData{$T$ a tree component of the network, $(\beta, d, \zcut, b)$}
    \KwResult{Updated values for $(\beta, d, \zcut, b)$}
    \For{each node $v\in T$ in the topological order of $T$}
        {
        \For{each class $c$ in parallel}
            {
            $b_c^v \gets \plpb_{v, c}(\beta_c, d_c, \zcuti, b_c)$\;
            }
        $\beta^v \gets \drrb_v(\beta^v, b^v)$\;
        \For{each class $c$ in parallel}
            {
            \For{$e \in \mathrm{In}_c(v)$ in parallel}
                {
                $b_c^e \gets \plpb_{e, c}(\beta_c, d_c, \zcuti, b_c)$\;
                }
            $d_c^v \gets \pd_{v,c}( \beta^v_c,b_c)$\; 
            }
        }
    \textbf{return} $(\beta, d, \zcut, b)$
    \caption{$\textrm{DRR}^{\textrm{tree}}_{T}$: DRR Analysis of a tree}
    \label{algo:tree-par}
\end{algorithm}

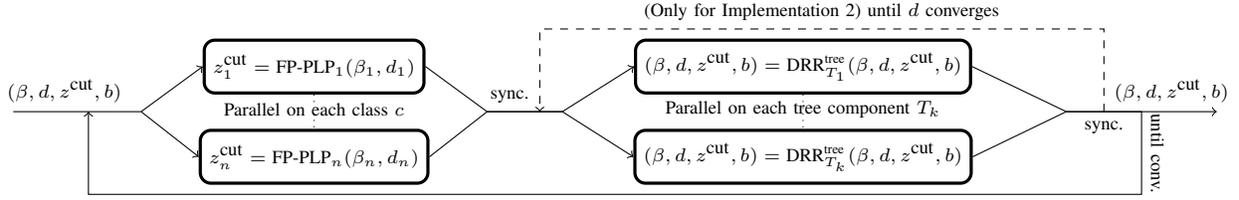
\begin{figure*}[htbp]
	\centering
	%  \title={.}
	%\begin{subfigure}[b]{0.85\textwidth}
	%	\centering
	%	\input{RTSS2022/TON_v1/Figures/version1-tikz}
		% include first image
		%\includegraphics[width=\linewidth]{RTSS2022/TON_v1/Figures/v1_para.pdf}
	%	\caption{Implementation 1: alternating between the $\plpcut$ and $\mathrm{DRR}^{\mathrm{tree}/\!\!/}_{T}$ blocks.}
	%\end{subfigure}
	%\hfill
	%\begin{subfigure}[b]{0.85\textwidth}
	%	\centering
		% include first image
			\tikzset{
fnc/.style={
	rectangle,
	rounded corners,
	draw=black, very thick,
	minimum height=2em,
	text centered,
	fill=white,
	font=\scriptsize}}
\begin{tikzpicture}
    \node[fnc] (fpplp1) at (2, 0.6) {$\zcut_1 = \plpcut_1(\beta_1, d_1)$}; 
    \node[fnc] (fpplpn) at (2, -0.6) {$\zcut_n = \plpcut_n(\beta_n, d_n)$}; 
    \draw[dotted] (fpplp1) -- (fpplpn) node[pos=0.43, font=\scriptsize] {Parallel on each class $c$};
    % \draw[dotted] (fpplp1) -- (fpplpn) node[pos=0.43, font=\scriptsize, right] {Parallel} node[pos=0.57, font=\scriptsize, right]{on each class $c$};
    % \draw[dotted] (fpplp1) -- (fpplpn) node[pos=0.43, font=\scriptsize, right] {Parallel} node[pos=0.57, font=\scriptsize, right]{on each class $c$};
    
    \node[fnc] (tree1) at (8.5, 0.6) {$(\beta, d, \zcut, b) = \textrm{DRR}^{\textrm{tree}}_{T_1} (\beta, d, \zcut, b) $}; 
    \node[fnc] (treek) at (8.5, -0.6) {$(\beta, d, \zcut, b) = \textrm{DRR}^{\textrm{tree}}_{T_k} (\beta, d, \zcut, b)$}; 
        \draw[dotted] (tree1) -- (treek) node[pos=0.43, font=\scriptsize] {Parallel on each tree component $T_k$};
    % \draw[dotted] (tree1) -- (treek) node[pos=0.43, font=\scriptsize, right] {Parallel} node[pos=0.57, font=\scriptsize, right]{on each tree component $T_k$};
    
    \draw[->] (-2, 0) -- (-0.3,0) node[pos=0.4, font=\scriptsize, above] {$(\beta, d, \zcut, b)$}-- (fpplp1.west);
    \draw[->] (-0.3,0) -- (fpplpn.west);
    
    \draw[->] (fpplp1.east) -- (4.3, 0) -- (5.3, 0) node[pos=0.3, font=\scriptsize,, above] {sync.} -- (tree1.west);
    \draw[->] (fpplpn.east) -- (4.3, 0) -- (5.3, 0) -- (treek.west);
    
    \draw[->] (tree1.east) -- (12, 0) -- (13, 0) node[pos=0.5, font=\scriptsize, below] {sync.} -- (13, -1.1) node[pos=0.5, font=\scriptsize, sloped, above] {until conv.} -- (-1, -1.1) -- (-1, 0);
    \draw[->] (treek.east) -- (12, 0) -- (14, 0) node[pos=0.7, font=\scriptsize, above] {$(\beta, d, \zcut, b)$};
    
    \draw[->, dashed] (12.5, 0) -- (12.5, 1.1) -- (5, 1.1) node[pos=0.5, font=\scriptsize, above] {(Only for Implementation 2) until $d$ converges} -- (5, 0);
\end{tikzpicture}
% \begin{tikzpicture}
%     \node[fnc] (fpplp1) at (2, 1) {$\zcut_1 = \plpcut_1(\beta_1, d_1)$}; 
%     \node[fnc] (fpplpn) at (2, -1) {$\zcut_n = \plpcut_n(\beta_n, d_n)$}; 
%     \draw[dotted] (fpplp1) -- (fpplpn) node[pos=0.43, font=\scriptsize, right] {Parallel} node[pos=0.57, font=\scriptsize, right]{on each class $c$};
    
%     \node[fnc] (tree1) at (8.5, 1) {$(\beta, d, \zcut, b) = \textrm{DRR}^{\textrm{tree}}_{T_1} (\beta, d, \zcut, b) $}; 
%     \node[fnc] (treek) at (8.5, -1) {$(\beta, d, \zcut, b) = \textrm{DRR}^{\textrm{tree}}_{T_k} (\beta, d, \zcut, b)$}; 
%     \draw[dotted] (tree1) -- (treek) node[pos=0.43, font=\scriptsize, right] {Parallel} node[pos=0.57, font=\scriptsize, right]{on each tree component $T_k$};
    
%     \draw[->] (-2, 0) -- (-0.3,0) node[pos=0.4, font=\scriptsize, above] {$(\beta, d, \zcut, b)$}-- (fpplp1.west);
%     \draw[->] (-0.3,0) -- (fpplpn.west);
    
%     \draw[->] (fpplp1.east) -- (4.3, 0) -- (5.3, 0) node[pos=0.3, font=\scriptsize,, above] {sync.} -- (tree1.west);
%     \draw[->] (fpplpn.east) -- (4.3, 0) -- (5.3, 0) -- (treek.west);
    
%     \draw[->] (tree1.east) -- (12, 0) -- (13, 0) node[pos=0.5, font=\scriptsize, below] {sync.} -- (13, -2) node[pos=0.5, font=\scriptsize, sloped, above] {until convergence} -- (-1, -2) -- (-1, 0);
%     \draw[->] (treek.east) -- (12, 0) -- (14, 0) node[pos=0.7, font=\scriptsize, above] {$(\beta, d, \zcut, b)$};
    
%     \draw[->, dashed] (12.5, 0) -- (12.5, 2) -- (5, 2) node[pos=0.5, font=\scriptsize, below=-0.1cm] {(Only for Implementation 2) until $d$ converges} -- (5, 0);
    
% \end{tikzpicture}
		%\includegraphics[width=\linewidth]{RTSS2022/TON_v1/Figures/v2_para.pdf}
	%\caption{Implementation 2: convergence of $\mathrm{DRR}^{\mathrm{tree}/\!\!/}_{T}$ before the execution of $\plpcut$.}
	%\end{subfigure}
	\caption{\sffamily \small Two implementations of the refinement phases with parallelization and shared memory. Implementation 1 (without the dashed arrow): alternating between the $\plpcut$ and $\textrm{DRR}^{\textrm{tree}}_{T}$ blocks; Implementation 2 (with the dashed arrow): convergence of $\textrm{DRR}^{\textrm{tree}}_{T}$ before the execution of $\plpcut$.
	%suggested order to apply each refinement. In both version, first bounds on the burstiness of flows at cuts $\zcut$ are updated; then, each node is visited with respect to the topological order of the tree; DRR strict service curve and then delay bounds are updated. For version 2, at each node we iterate until the delay bounds converge. We then start over and update $\zcut$.
	}
	\label{fig:versions}
\end{figure*}
The two implementations are illustrated in Fig.~\ref{fig:versions}. The first implementation (without the dashed arrow) alternates between the two sets of blocks ($\plpcut$ and $\textrm{DRR}^{\textrm{tree}}_{T}$). Each set of blocks is started after synchronizing the previous set of blocks. 
We start by $\plpcut$ blocks because classes considered unstable in the initial phase (the TFA analysis outputs infinite bounds) might become stable after PLP analysis. %, as the stability region computed by PLP is larger than that computed by TFA. 
In contrast, $\textrm{DRR}^{\textrm{tree}}_{T}$ cannot improve the stability region. 

The second implementation tightens as much as possible the DRR service curves and per-node delays before executing again the first $\plpcut$ block (dashed arrow in Fig.~\ref{fig:versions}). The aim of this implementation is to execute the $\plpcut$ block less often as it is more time-consuming. Indeed, it requires solving a much larger LP than in the other block. Therefore, the second block is executed several times until convergence is reached on the delay bounds.

\subsection{Post-process Phase: Computing the End-to-end Delay} 
\label{sec:postPhase}
  When the refinement phase of Section \ref{sec:refinePhase} has converged, we proceed and compute end-to-end delay bound for each flow of interest. Strict service curves $\beta$ obtained are piece-wise linear and convex. As stated in Section~\ref{sec:drrservice}, delays can be  improved when considering the non-convex strict service curve described in  Fig.~\ref{fig:ncService}. Thus, we apply $\iplpd$ to compute end-to-end delay bounds.

	%\input{sot}
	%\input{method}
    
% !TeX root = mainTON.tex
 %\vspace{-0.2cm}
\section{Proofs}\label{sec:proof}

\subsection{Proof of Theorem \ref{thm:plpBack}}
\label{sec:proof-plpBack}
We first prove Lemma~\ref{lem:aggBacklog}: Consider a system $\mathbb{S}$ and a set of flows $F$ that traverses $\mathbb{S}$. Let $A_f$ (resp. $D_f$) denote the cumulative arrival (resp. departure) of flow $f\in F$. Let $A = \sum_{f\in F} A_f$ (resp. $D = \sum_{f \in F} D_f$) denote arrival (resp. departure) of the aggregate of flows of interest. Assume that: (A1) System $\mathbb{S}$ is causal, i.e., $\forall (A,D) \in \mathbb{S},~ A \geq D$ and $D(t)$ only depends on $A(s)_{s\leq t}$; (A2) The departure $D$ of system $\mathbb{S}$ is continuous; (A3) Every flow of interest $f\in F$ has a token-bucket arrival curve $\alpha_f $ with rate $r_f$ and burst $b_f$; also, $\alpha_f$  is the only constraint on the arrival of the flow of interest, i.e., every cumulative arrival $A_f$ constrained by $\alpha_f$ is a possible arrival for this flow; (A4) $\mathcal{B}$ is a backlogged bound for every possible arrival and departure of the aggregate of  flows of interest that belongs to system $\mathbb{S}$, i.e., $\forall (A,D) \in \mathbb{S}$, we have $A - D \leq \mathcal{B}$.

% To prove 1), we prove Lemma~\ref{lem:aggBacklog}: Consider a system $\mathbb{S}$ and a set of flows $F$ that traverses $\mathbb{S}$. Let $A_f$ (resp. $D_f$) denote the cumulative arrival (resp. departure) of flow $f\in F$. Let $A = \sum_{f\in F} A_f$ (resp. $D = \sum_{f \in F} D_f$) denote arrival (resp. departure) of the aggregate of flows of interest. Assume that: (A1) System $\mathbb{S}$ is causal; (A2) The departure $D$ of system $\mathbb{S}$ is continuous; (A3) Every flow of interest $f\in F$ has a token-bucket arrival curve $\alpha_f $ with rate $r_f$ and burst $b_f$, and
% % also, $\alpha$  is the only constraints on the arrival of the flow of interest, i.e.,
% every cumulative arrival $A_f$ constrained by $\alpha_f$ is a possible arrival for this flow; (A4) $\mathcal{B}$ is a backlogged bound for every possible arrival and departure of the aggregate of  flows of interest that belongs to system $\mathbb{S}$, i.e., $\forall (A,D) \in \mathbb{S}$, we have $A - D \leq \mathcal{B}$.
%   \vspace{-0.15cm}
\begin{lemma}\label{lem:aggBacklog}
	Assume the assumptions (A1)-(A4). Then, the departure of the aggregate of the set of flows of interest $F$ is constrained by  a token-bucket arrival curve with rate $r$ and burst $\mathcal{B}$ where $r = \sum_{f\in F} r_f$, i.e., $\gamma_{r,\mathcal{B}}$.   
\end{lemma}
% \anne{This result already appears in~\cite{Bou19}, unproved, and in a slightly more restrictive case.} 
%  \vspace{-0.15cm}
Note that in a specific case where we have only one flow of interest, assumptions (A1)-(A4) are satisfied by those of Theorem 4 of \cite{plp}, and both theorems  give exactly the same result. Also, this result already appears in Corollary 2 of ~\cite{Bou19}, unproved, and in a slightly more restrictive case.
%  \vspace{-0.15cm}
\begin{proof}
	% The proof is similar to a first part of the proof of Theorem 4 of \cite{plp}. 
	Fix $(A_f,D_f)\in \mathbb{S}$ for every $f\in F$ and $s\leq t$. Let $r = \sum_{f\in F} r$ and $b = \sum_{f\in F} b_f$. We prove that $ D(t) - D(s) \leq \gamma_{r,\mathcal{B}}(t-s) = \mathcal{B} + r(t-s).$
	% \begin{equation}\label{eq:1}
	%     D(t) - D(s) \leq \gamma_{r,\mathcal{B}}(t-s) = \mathcal{B} + r(t-s).
	% \end{equation}
	
	We first prove that $ A(s) - D(s) + H \leq \mathcal{B}$ with $H = b - \bar{b}(s)$ and $\bar{b}(s) \isdef \sup_{u\leq s}\{A(s) - A(u) - r(s-u) \}$. Note that $\bar{b}(s)$ is the bucket size of the token-bucket $\gamma_{r,b}$ at time $s$, so $H\geq 0$.
	% \begin{equation}\label{eq:2}
	%     A(s) - D(s) + H \leq \mathcal{B}
	% \end{equation}
	% with
	% \begin{align}
	%     H &= b - \bar{b}(s)\\
	%     \bar{b}(s) &\isdef \sup_{u\leq s}\{A(s) - A(u) - r(s-u) \}
	% \end{align}
	% Note that $\bar{b}(s)$ is the bucket size of the token-bucket $\gamma_{r,b}$ at time $s$.
	Define $A'$ as follows: $A'(u) = A(u)$ for $u \leq s$ and $A'(u) = A(s) + H + r(u-s)$ for $u > s$. Observe that $A'$  is constrained by $\gamma_{r,b}$ and $A' \geq A$. Hence, by (A3), $A'$ is a possible arrival in $\mathbb{S}$: there exists $D'$ such that $(A',D')\in \mathbb{S}$. Hence, by (A4), we have $ A'(s^+) - D'(s^+) \leq \mathcal{B}$
	% \begin{equation}\label{eq:3}
	%     A'(s^+) - D'(s^+) \leq \mathcal{B}
	% \end{equation}
	where $f(x^+) = \lim_{y\rightarrow x, y > x}f(y)$. As $A'(u) = A(u)$ for $u \leq s$ and by (A1), we have $D'(u) = D(u)$ for $u \leq s$; also, by (A2), $D'(s^+) = D'(s)$.  By combining this and  $A'(s^+) = A(s) + H$, we obtain  $A(s) - D(s) + H \leq \mathcal{B}$. 
	
	To conclude, we notice that $D(t)\overset{(A1)}{\leq } A(t) \leq A'(t) = A(s) + H + r(t-s)$, so $D(t) - D(s) \leq A(s) + H + r(t-s) - D(s) \leq \mathcal{B} + r(t-s)$.
	%We now proceed the proof as follows: By A1, $D(t) - D(s) \leq A(t) - D(s)$. As $A' \geq A$, we have $D(t) - D(s) \leq A'(t) - D(s)$.  By plugging in the definition of $A'(t)$, we have $D(t) - D(s) \leq A(s) - D(s) + H + r(t-s)$. Then, suing the inequality proved in the previous paragraph, we have $D(t) - D(s) \leq \mathcal{B} + r(t-s)$. Hence \eqref{eq:1} is shown. 
	% We now proceed to prove \eqref{eq:1} as follows:
	% \begin{align}
	%     D(t) - D(s) &\leq A(t) - D(s) \textrm{ (by A1)}\\
	%                 &\leq A'(t) - D(s) \\
	%                 & =A(s) - D(s) + H + r(t-s) \\
	%                 & \leq \mathcal{B} + r(t-s) \textrm{ (by \eqref{eq:3})}
	% \end{align} 
	% Hence \eqref{eq:1} is shown. 
\end{proof}

% Lemma \ref{lem:aggBacklog} proves 1), i.e., that a backlog bound for an aggregate of some flows is a bound on the output burstiness for the aggregate of these flows. It remains to show 2), i.e., that $\plpb_{g,c}$ computes a backlog bound for an aggregate of of flows in $F$. 
%  %\vspace{-0.2cm}
% Lemma \ref{lem:aggBacklog} proves 1). It remains to show 2). 
%  \vspace{-0.2cm}
Theorem 5 of \cite{plp} proves the correctness of PLP. The proof consists in showing that any trajectory scenario of the network is a feasible solution of the PLP; given a trajectory scenario, the variables of the PLP (time variable and process variables) are extracted from the scenario such that all constraints are satisfied. We only add some arrival curve constraints for flows $f\in F$. In \cite{plp} it is shown that these constraints are correct for a single flow and hence the same proof can be used for each flow $f\in F$. In particular, $\mathbf{Ft}_{v_0, 0}^{v_f, f}$ (resp. $\mathbf{Ft}_{v_0, 0}^{v_0, f}$) represents the arrival process (departure process from the system) of flow $f$ at time $\tvar{v_0}{\scriptsize{0}}$. No other constraint is modified. Only the objective function changes and maximizes the quantity of data of the flows of interests at time $\tvar{v_0}{\scriptsize{0}}$:  $\sum_{f\in F}\mathbf{Ft}_{v_0, 0}^{v_f, f} - \mathbf{Ft}_{v_0, 0}^{v_0, f}$ is a backlog bound the aggregate flows and $\plpb_{g,c}$ computes a backlog bound for an aggregate of flows when $F$ is chosen to be the set of flows of class $c$ traversing node  or edge $g$. By Lemma \ref{lem:aggBacklog} this is a bound on the output burstiness. \qed

% Theorem \ref{thm:backlog} enables us to use PLP and compute a bound on on the burstiness of the aggregate of some flows: We use the same PLP used to compute a backlog bound for a single flow yet the objective is now for the aggregate, i.e., a backlog bound for aggregate of flows in the network. Recall  the notations of Section \ref{sec:iplp}. Consider an edge $E \subseteq \textrm{Out}(v)$; $v$ is the root of the tree.  Then, add the following constraints for each $f$: $\forall k \in \{0, \ldots, \depth(\pi_f(1)) \}, \fvar{\pi_f(1)}{f}{v+1}{0} - \fvar{\pi_f(1)}{f}{\pi_f(1)}{k} \leq b_f + r_f(\tvar{v+1}{0} - \tvar{\pi_f(1)}{k})$. Then the objevtice is to maximize $\sum_f \lp\fvar{\pi_f(1)}{f}{v+1}{0} - \fvar{v+1}{f}{v+1}{0} \rp $.

% Then, by Theorem \ref{thm:backlog}, an output arrival curve an be computed for the aggregate.

% \vspace{-0.3cm}
%  %%\vspace{-0.2cm}
\subsection{Proof of Theorem \ref{thm:iplp}}
\label{sec:proof-iplp}
%%\vspace{-0.1cm}
We proceed as explained in the last paragraph above and prove that the constraints we added are correct. 
% of \cite{plp} proves the correctness of PLP. The proof is consist of showing that any trajectory scenario of the network is a feasible solution of the PLP; given a trajectory scenario, the variables of the PLP (time variable and process variables) are extracted from the scenario such that all constraints are satisfied. As we only add some constraints for the service curves at each node and we do not touch other constraints we only have to check whether the constrains that we added are correct. 
%Specifically,  the four constraints of \eqref{eq:nc1}-\eqref{eq:nc4} corresponds to strict service curve of Fig.~\ref{fig:ncService}.
Specifically, we only need to prove that a solution of the linear problem satisfies (with the notations previously introduced) $\mathbf{At}_u^v - \mathbf{At}_v^v \geq \max(\beta_c^{v}, \beta_c^{nc, v})(\tvar{u}{\scriptsize\depth(u)} - \tvar{v}{\scriptsize\depth(v)}).$
% \begin{equation}
%     \label{eq:iplp}
%     \mathbf{At}_u - \mathbf{At}_v \geq \max(\beta_c^{v}, \beta_c^{nc, v})(\tvar{u}{\scriptsize\depth(u)} - \tvar{v}{\scriptsize\depth(v)}).
% \end{equation}
From~\eqref{eq:sc1} and~\eqref{eq:sc2}, $\mathbf{At}_u^v - \mathbf{At}_v^v \geq \beta_c^{v}(\tvar{u}{\scriptsize\depth(u)} - \tvar{v}{\scriptsize\depth(v)})$ holds, so we just need to focus on $\beta_c^{nc, v}$, and distinguish two cases, depending on the value of $\mathbf{b}_v$: This follows from the cases explained above~\eqref{eq:nc1} and~\eqref{eq:nc3} and the fact that $M$ is large enough.
% if $\mathbf{b}_v = 1$, then ~\eqref{eq:nc1} and~\eqref{eq:nc3} are dummy (because $M$ is large enough) and from ~\eqref{eq:nc2}, $\mathbf{At}_u^v - \mathbf{At}_v^v \geq q_c^v \geq \beta_c^{nc, v} (\tvar{u}{\scriptsize\depth(u)} - \tvar{v}{\scriptsize\depth(v)})$. Otherwise, $\mathbf{b}_v = 0$ and the (non dummy) constraints from ~\eqref{eq:nc1} and~\eqref{eq:nc3} enforce that  $\mathbf{At}_u^v - \mathbf{At}_v^v \leq q_c^v$ and under this constraint, $\mathbf{At}_u^v - \mathbf{At}_v^v \geq R^v(\tvar{u}{\scriptsize\depth(u)} - \tvar{v}{\scriptsize\depth(v)} - T_1) \geq  \beta^{nc, v}_c (\tvar{u}{\scriptsize\depth(u)} - \tvar{v}{\scriptsize\depth(v)})$. 
Hence, it finishes the proof of 1). As iPLP has more constraints than PLP, the output of iPLP is less than or equal to that of PLP, which concludes 2). \qed

%\subsection{Proof of Theorem \ref{thm:tfa}}
%\hossein{Should we combine proof of Theorem \ref{thm:tfa} and \ref{thm:genMethod}? Theorem \ref{thm:tfa}  is a simpler application of the other where we have two functions $h$ and it is sequential.}
%\anne{That was my first plan before I saw the proof of Theorem 3 in the text. If you want I can try a combination}
%  \vspace{-0.2cm}
\subsection{Proof of Theorems~\ref{thm:tfa} and~\ref{thm:genMethod}}
\label{sec:proof-shared-memory}
% We first prove Theorem \ref{thm:genMethod} and the proof of Theorem \ref{thm:tfa} is obtained in a similar way.

% Before proving the theorems, let us first precise the general behavior of a shared memory in this paper.
%\vspace{-0.1cm}
%%%%%%% Shorter version:
We use the following lemma. We assume a finite set of isotone mappings $\mathcal{H}$,  in $\mathcal{C} \to \mathcal{C}$ with $\mathcal{C} \subseteq (\Reals^+\cup\{+\infty\})^I$. An \emph{execution} of $\mathcal{H}$ is $(h_k, s_k, t_k, u_k)_{k\geq 1}$ such that: $\forall k\geq 1$, (C1)  $h_k \in\mathcal{H}$; (C2) $0< s_k\leq t_k < u_k$; (C3)  $u_{k+1} > u_{k}$; %$s_{k+1} > s_{k}$; 
(C4) $\forall k'\neq k$, $(t_k, u_k] \mand (t_{k'}, u_{k'}]$ are disjoint.  We also assume that (C5) Each function $h\in\mathcal{H}$ is executed   infinitely many  times; (C6) $\lim_{k\to\infty} s_k = \infty$. In the description of Section~\ref{sec:refinePhase}, $s_k$, $t_k$ and $u_k$ respectively correspond to the reading time, locking and unlocking times of the write operation in the execution of the $k$-th refinement (by order of completion times). Let $x_0\in\mathcal{C}$ be the state of the memory at time 0. Given an execution $(h_k, s_k, t_k, u_k)_{k\geq 1}$, the state of the memory $x(t)$ evolves with time $t$ as follows: $x(t)$ is piece-wise constant, right-continuous; it is modified at times $u_k$, $k\geq 1$, and $x(u_k) = \min(x(t_k), h_k(x(s_k))$. 
\begin{lemma}\label{lem:fix}
	%Let $x_0\in\mathcal{C}$. Then there exists $x^*$ such that for all executions of $\mathcal{H}$ satisfying (H1) and (H2), 
	%timed sequence satisfying assumption (A), 
	%the memory state from $x_0$ satisfies $\lim_{t\to\infty} x^*$. \hossein{
	There exists $x^*$ such that for all executions of $\mathcal{H}$ as explained above,   $\lim_{t\to\infty} x(t) = x^*$.
	%}
	% \anne{keep?} Moreover, $x^*$ is a fixed point for all function $h\in\mathcal{H}$: $h(x^*) = x^*$. 
\end{lemma}
%  \vspace{-0.25cm}
\begin{proof}
	Let us first prove that $\lim_{t\to\infty} x(t)$ exists given $(h_k, s_k, t_k, u_k)_{k\geq 1}$. %, $\lim_{t\to\infty} x(t)$ exists. %As $x$ is piece-wise constant with jumps at $(u_k)_{k\geq 1}$, 
	%Let $\varphi: \mathbb{N} \setminus \{0\} \to \mathbb{N} \setminus \{0\} $ be the one-to-one function such that $(u_{\varphi(k)})_{k\geq 1}$ is non-decreasing (it is the ordered sequence of jumps of $x$). 
	From (C2) and (C4),  $u_{k} \leq t_{k+1} < u_{k+1}$ holds, so $x(t_{k+1}) = x(u_{k})$. Then, $x(u_{k+1}) = \min(h_{k+1}(x(s_{k+1})), x(u_{k})) \leq x(u_{k})$, and $(x(u_{k}))_{k\geq 1}$ is a non-increasing sequence in $(\Reals \cup \{+\infty\})^I$, hence converges. As $\lim_{k\to\infty} u_{k} = +\infty$, $x(t)$  converges. 
	
	Second, we prove that the limit of $x$ depends  only on the initial value $x_0$.
	%does not depend on the execution of $\mathcal{H}$  given the initial value $x_0$. 
	Consider two executions $(h_k, s_k, t_k, u_k)_{k\geq 1}$, and $(h'_k, s'_k, t'_k, u'_k)_{k\geq 1}$, and the state of their respective shared memory $x$ and $x'$, with respective limits $x^*$ and $x'{}^*$. %Consider function $\varphi$ as defined in the previous paragraph, with $\varphi(0) = 0$, and 
	Set $u_0 = 0$ and define $\varphi$ the following way: Set $u'_0 = 0$ and $\varphi(0) = 0$. For all $k\geq 1$, define $\varphi(k) = \min\{k'\geq 1,~h'_{k'} =h_{k} \text{ and } s'_{k'} \geq u'_{\varphi(k-1)}\}$. In particular, by (C2) $u'_{\varphi(k-1)} \leq s'_{\varphi(k)} < u'_{\varphi(k)}$, so by (C3) $(\varphi(k))_{k\geq 1}$ is (strictly) increasing. 
	
	Let us now show by induction that $x(u_{k}) \geq x'(u'_{\varphi(k)})$ for all $k\geq 0$. The base case holds since $x(0) = x'(0) = x_0$. Let us now assume that $x(u_{k}) \geq x'(u'_{\varphi(k)})$. 
	
	On the one hand, $x(u_{k+1}) = \min(h_{k+1}(x(s_{k+1})), x(u_{k}))$. By induction hypothesis,  $x(u_{k}) \geq  x'(u'_{\varphi(k)})$. By construction $h_{k+1} = h'_{\varphi(k+1)}$ and by (C2) $s_{k+1} < u_{k+1}$, so $x(s_{k+1}) \geq x(u_{k})\geq x'(u'_{\varphi(k)})$.
	Then, it follows that $x(u_{k+1}) \geq \min(h'_{\varphi(k+1)}(x'(u'_{\varphi(k)})), x'(u'_{\varphi(k)}))$.
	%= h'_{\varphi'(k+1)}(x'_{u_{\varphi'(k)}})$.
	%
	
	On the other hand, $x'(u'_{\varphi(k+1)}) = \min (h'_{\varphi(k+1)}(x'(s'_{\varphi(k+1)})), x'(t'_{\varphi(k+1)}))$. 
	By construction of $\varphi$,  $t'_{\varphi(k+1)} \geq s'_{\varphi(k+1)} \geq u'_{\varphi(k)}$, so $x'(t'_{\varphi(k+1)}) \leq x'(u'_{\varphi(k)})$. 
	%As $x'$ is non-increasing, 
	We finally obtain $x'(u'_{\varphi(k+1)}) \leq  \min(h'_{\varphi(k+1)}(x'(u'_{\varphi(k)})), x'(u'_{\varphi(k)}))$. %= %h'_{\varphi'(k+1)}(x'_{u'_{\varphi'(k)}})$.
	Hence,  $x(u_{k+1}) \geq x'(u'_{\varphi(k+1)})$, which proves the induction step. 
	
	Therefore, $x^* = \lim_{t\to\infty}x(t) = \lim_{k\to\infty} x(u_{k}) \geq \lim_{k\to\infty} x'(u'_{\varphi(k)}) = x'{}^*$.
	Inverting the roles of the two executions finishes the proof.
	%        Inverting the roles of $(h_k, s_k, t_k, u_k)_{k\geq 1}$ and $(h'_k, s'_k, t'_k, u'_k)_{k\geq 1}$ leads to $x^* \leq x'{}^*$ and hence $x^*=x'{}^*$. 
	%
	% We now prove that for each function $h\in\mathcal{H}$,  $x^* = h(x^*)$. First, we know by definition of $h$, that $x^* \geq h(x^*)$. Second, fix an execution $(h_k, s_k, t_k, u_k)$.  Since $h$ appears infinitely often in this execution, one can define $\psi$ the sequence of the occurrences of $h$: $h_{\psi(k)}$ is the $k$-th occurrence of $h$ in the execution. Then $x(u_{\psi(k)}) = \min(x(t_{\psi(k)}), h(x(s_{\psi(k)}))) \leq  h(x(s_{\psi(k)}))$. Since $(s_k)$, $(t_k)$ and $(u_k)$ all goes to $\infty$ as $k$ increases, we have  
	% $x^* \leq h(x^*)$ (\anne{We should have a continuity hypothesis somewhere to conclude that $\lim h(x_k) = h(x^*)$. Is this used somewhere?})
\end{proof}
\begin{proof}[Proof of Theorem~\ref{thm:genMethod}]
	The shared memory contains  non-negative numbers (burstiness bounds and per-node delays) and piece-wise linear convex functions (the DRR strict service curves). First, observe that the piece-wise linear convex functions we deal with can be described by a finite set of elements of $\Reals^+ \cup \{+\infty\}$.
	As explained in Section~\ref{sec:drr}, the result of $\drrb(\beta^v, b^v)$ is the maximum of a finite number of rate-latency functions, whose rates are in a finite set, depending on fixed parameters (the arrival rates of flows) and the latencies are linearly decreasing with the output burstiness bounds $b^v$. 
	Hence, the shared memory can be expressed as a family of $x = (T, d, \zcut, b)$.
	
	%First, we observe the following: Consider the DRR strict service curve refinement of Corollary \ref{cor:serviceBurst} ${\beta^v}' = \drrb_v \lp b^* ,\beta^v \rp$:  ${\beta^v}'$ is equal to maximum of $\beta^v $ and, as explained after Corollary \ref{cor:serviceBurst},  finite number of rate-latency functions, where the latency part is a linear, increasing function of $b^*$, the output burstiness. This means each time $\drrb_v$ is visited, only finite number of latencies are improved. Hence, what we improve in phase 1 are triple of $\lp T,d,z \rp $, where $T$ is the collection of latencies. Then, applying each refinement to $\lp T,d,z \rp $ result in $\lp T',d',z' \rp $, where $T' \leq T$, $d' \leq d$, and $z' \leq z$.

	\subsubsection{Validity} Assuming valid initial bounds, the four types of functions used, $\plpb$, $\plpcut$, $\drrb$ and $\pd$ (defining set $\mathcal{H}$), provide valid bounds, proved in the literature.
	%The initial value is the output of the initial phase proved to be valid. Then four types of functions are used: $\plpb$, $\plpcut$, $\drrb$ and $\mathcal{D}$ (defining set $\mathcal{H}$). For all these operations have been proved valid in the literature if they are used on valid input. 
	Moreover, if $(T, d, \zcut, b)$ and $(T', d', z'{}^{\textrm{cut}}, b')$ are valid bounds, then $(T\land T', d\land d', \zcut\land z'{}^{\textrm{cut}}, b\land b')$ are also valid bounds (where $\land$ is the minimum operation). This is straightforward for the per-node delays and burstiness bounds. For the service curves, it is enough to notice that the maximum of two strict service curves for a  node is also a strict service curve for that node. So if $x(t_k)$ and $h_k(x(s_k))$ represent valid bounds, $\min(x(t_k), h_k(x(s_k)))$ are also valid bounds, so $x(t)$ always represents valid bounds.

	\subsubsection{Convergence} We then set $\mathcal{C}$ as the set of valid parameters for the problem
	and apply Lemma~\ref{lem:fix} where
	$s_k$, $t_k$ and $u_k$ respectively correspond to the reading time, locking and unlocking times of the write operation. (C1)--(C3) hold by definition and (C4) because of the lock operation. (C5) follows from (H). Furthermore, since there is a finite number of workers, (H) also implies that every execution completes except for at most a finite number, which implies (C6).
	%\jylb{say why (C4)-(C6) hold}. %whatever the order of the operations, provided each refinement is applied infinitely often, the shared memory converges to a values independent of the order of operations. 
\end{proof}
% \vspace{-0.25cm}
%
\begin{proof}[Proof of Theorem \ref{thm:tfa}]\label{sec:proof-tfa}
	Consider lines 2-6 of Algorithm~\ref{alg:initPahse}, and assume an infinite loop. The algorithm is sequential and hence is %can be seen as 
	a specific case of the shared memory computing, where updates are one after the other, i.e., $\forall k, s_{k+1} > u_{k}$. The variables $x$ is the collection $(T, d, z)$ and the initial value is $+\infty$ at each coordinate except two latencies per node and class (from line~\ref{line:betaCDM}). 
	Two types of functions are used: $\tfa$ and  $\drri$ (defining set $\mathcal{H}$). Note that decreasing delays and bursts decreases the latencies involved in the DRR service curves and conversely, so at each step $x(t_k)=x(s_k) \geq h_k(x(s_k))$ and $x(u_k) = h_k(x(s_k))$, which is exactly how the algorithm is updated. As a consequence, Lemma~\ref{lem:fix} can be applied: $x(t)$ converges in $\Reals^+ \cup\{+\infty\}$. 
	%In particular, every coordinate either becomes finite and converges in $\Reals^+$, either remains $+\infty$. 
	%          Hence, the proof given above holds for this theorem as well; therefore, a the limit exists and bounds are valid (if finite).
\end{proof}

\section{Numerical Evaluation} \label{sec:numEval}

\begin{figure}[t]
	\centering
	%	\title={.}
	\includegraphics[width=0.65\linewidth]{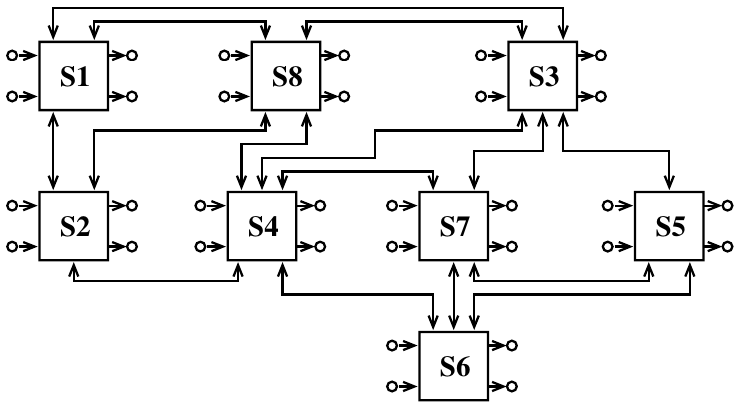}
	\caption{\sffamily \small Industrial-sized network topology. Figure from \cite{1647738}.}
	\label{fig:indusNet}
\end{figure}

\begin{figure}[b]
    \centering
    \includegraphics[width=\columnwidth]{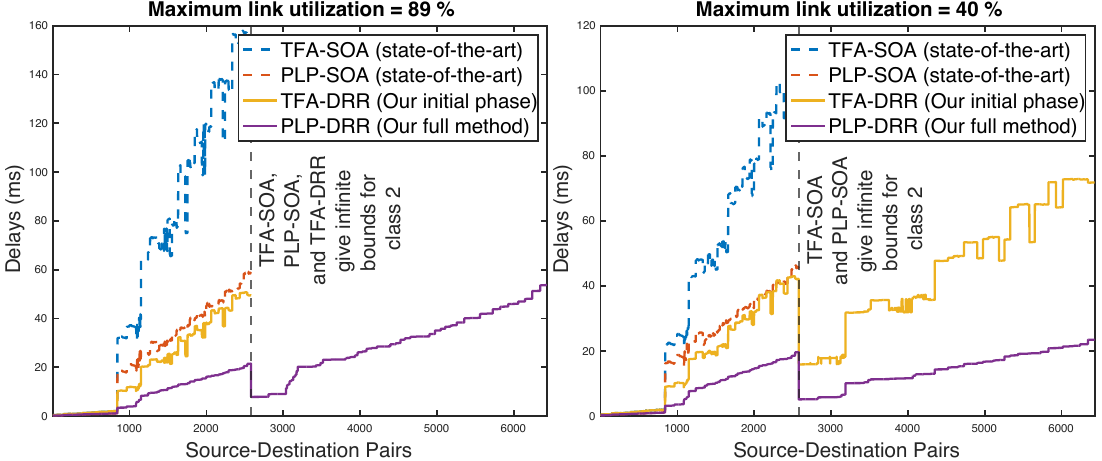} 
    \caption{Delay bounds obtained by our methods, TFA-DRR and PLP-DRR (plain plots), and the state-of-the-art, TFA-SOA and PLP-SOA (dashed plots). TFA-SOA and PLP-SOA use DRR strict service curves in the degraded operational mode. TFA-SOA and PLP-SOA both provide infinite bounds for class $2$ even when the maximum link utilization is $40 \%$. Source-destination paths are ordered by values of our full method, and finite delays for other methods are shown first.}%
    \label{fig:sot}%
\end{figure}

% \begin{figure}[htbp]
%     \centering
%     \subfloat{{\includegraphics[width=0.4\columnwidth]{TON_v1/Figures/plp_tfa_sot_40.pdf} }}%
%     \qquad
%     \subfloat{{\includegraphics[width=0.4\columnwidth]{TON_v1/Figures/plp_tfa_sot_89.pdf} }}%
%     \caption{Compression of bounds obtained by our methods, TFA-DRR and PLP-DRR (plain plots), to the state-of-the-art, TFA-SOT and PLP-SOT (dashed plots). TFA-SOT and PLP-SOT use DRR strict service curves in the degraded operational mode (i.e., service curves that only depend on the assigned quentua and maximum packet sizes and have no assumption on the interfering traffic. Our methods significantly improve the obtained bounds and give a larger stability region:  TFA-SOT and PLP-SOT both provide infinite bounds for class $2$ even when the maximum link utilization is $40 \%$.)
%     }%
%     \label{fig:sot}%
% \end{figure}

\begin{figure*}[htbp]
	\centering
	%	\title={.}
	\includegraphics[width=\linewidth]{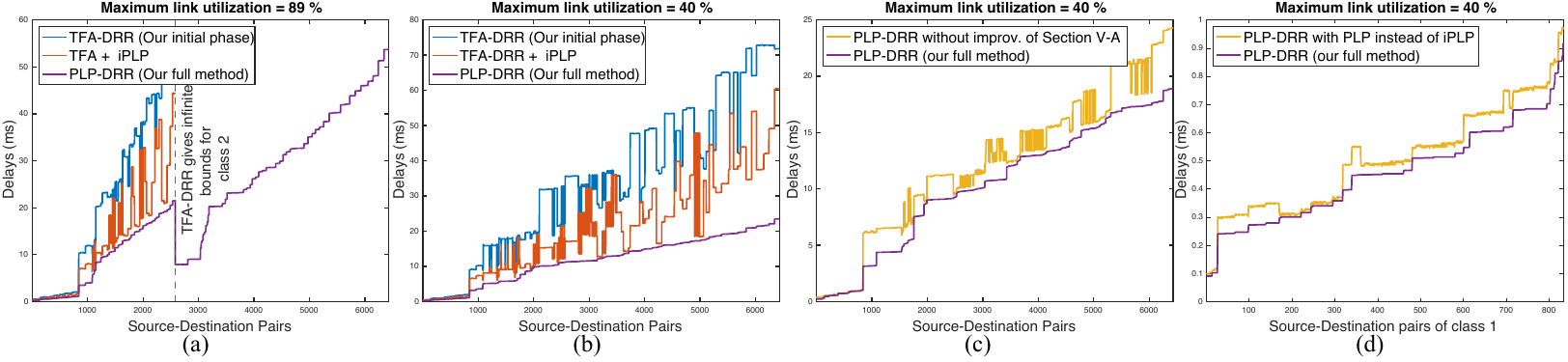}
	\caption{\sffamily \small 
	Delay bounds of our method compared to alternative methods. (a-b): comparison with 1)-2). At link utilization $89 \%$, TFA-DRR analysis gives infinite delays for class 2. (c-d): effect of the two PLP improvements of Section~\ref{sec:2imp}. Source-destination paths are ordered by values of our full method, except in (a) where finite delays for TFA-DRR are shown first.
	%(a),(b): Delay bounds of our method compared to alternatives that uses some elements of our method. When the maximum link utilization reaches $89 \%$, TFA analysis of class $2$ gives infinite bounds. \anne{However, our method gives finite bounds for any maximum link utilization below $100 \%$. this is already written in the main text and has nothing to do with the figure} (c),(d): Showing the effect of the two improvements to PLP in Section~\ref{sec:2imp}. At higher link utilizations the improvement are larger than on the figure. Source-destination paths are ordered by values of our full method \hossein{(for (a) paths with infinite TFA bounds are separately ordered.)} \jylb{Also for panel (a)?
	}%. This is a comparison between our full method and our full method without out PLP improvement of per-aggregate bound on burstiness (Theorem \ref{thm:plpBack}). At higher link utilization we have a better improvements.  (b) When delay bounds are small, non-convex part of DRR strict service curves bring an improvement. This is captured by our iPLP. Source-destination paths are ordered by values of our full method. Source-destination paths are ordered by values of our full method.
	\label{fig:numFull}
\end{figure*}

We use the network of Fig.~\ref{fig:indusNet}, a test configuration provided by Airbus in \cite{Grieu-line-shaping}. The industrial-sized case study of \cite{drr_rtas} is based on this network in \cite{1647738}.  It includes 96 end-systems, 8 switches, 984 flows, and 6412 possible paths. The rate of the links are equal to $R =1$ Gb/s, and every switch $S_i$ has a switching latency equal to $16 \mu$s. Every switch has 6 input and 6 output end-systems. There are three classes of flows: (1) critical, (2) multimedia and (3) best-effort. Worst-case delay bounds are required for classes 1 and 2 only. There is one DRR scheduler at every switch output port with $n=3$ classes (see Table~\ref{tab:networkConfig} for details).
% At every DRR scheduler, the quanta are 3070 bytes for the critical class, 1535~bytes for the multimedia and best-effort classes.
% %class, and 1535~bytes for the best-effort class.
% 128 multicast flows, with 834 destinations, are critical; they have a maximum packet-size of 150 bytes. 500 multicast flows, with 3845 destinations, are multimedia;  %and their class has a quantum equal to 1535 bytes
% they have a maximum packet-size of 500 bytes. 266 multicast flows, with 1733 destinations, are best-effort; they have a maximum packet-size of 1535 bytes. 
For every flow, the path from the source to a destination can traverse at most $4$ switches. In \cite{drr_rtas}, as their method only applies to feed-forward networks, flow paths are chosen randomly with the constraint that graphs induced by flows are feed-forward. In this paper, we removed this restriction and, as the network has much redundancy, this automatically generates induced flow graphs with cyclic dependencies, and is more representative of a realistic deployment. We obtain different network utilization factors by varying the minimum packet inter-arrival times. We consider two modes: when the network is lightly loaded with a maximum link utilization of 40\%, and when the network is highly loaded with a maximum link utilization close to 100\%.

% \bleu{(it is ....  and for the maximum link utilization respectively 40\% and 89\%)}.

\begin{table}[htbp]
	\centering
	\caption{\sffamily \small Traffic Characterization}
	\scalebox{0.73}{
		\begin{tabular} {|c|| c c c |}
		\hline		
		Traffic Classes & Number of Flows &  Assigned Quantum (bytes) &  Maximum Packet size (bytes)\\
		\hline
        Critical &$834$ &$3070$ &$150$   \\
		\hline
		Multimedia  &$3845$& $1535$ & $500$  \\
		\hline
  	Best Efforts &$1733$& $1535$ & $1535$  \\
		\hline
	\end{tabular}
	}
	\label{tab:networkConfig}
\end{table}

As of today, the only methods that compute bounds on the worst-case delay of DRR networks with cyclic dependencies are TFA-SOA (see Section~\ref{sec:tfa}) and PLP-SOA (see Section~\ref{sec:plp}) where they use DRR strict service curves in degraded operational mode (see Section~\ref{sec:drrservice_dm}). As illustrated in Fig.~\ref{fig:sot}, our methods TFA-DRR and PLP-DRR significantly improve delay bounds and provide larger stability regions.

 As our method, PLP-DRR, contains a number of improvements, we perform a numerical analysis to evaluate whether each of our improvement is useful or not. We compare our full method, PLP-DRR, to the following four alternatives:

1) TFA-DRR: we apply the initial phase only; this is the best that can be obtained with TFA and DRR service curves.

2) TFA-DRR + iPLP: we apply the initial phase and the post-process phase but do not apply the refinement phase; this shows the effect of the refinement phase.

3) Full method but %with one difference:  
we do not use our improvement in Section~\ref{sec:plpBack} to compute per-aggregate bounds; instead, we sum the burstiness bounds of every individual flow obtained using PLP; this shows the effect of this PLP improvement. 

4) Full method but with PLP instead of iPLP in the post-process phase; this shows the effect of this PLP improvement. 

% \begin{table}[htbp]
% 	\centering
% 	\caption{\sffamily \small \bleu{Run-times of initial and Refinement Phase}}
% 	\scalebox{0.7}{
% 		\begin{tabular} {|c|| c c c |}
% 		\hline		
% 		Maximum Link Utilization &Initial Phase (TFA) &  Refinement Phase Version 1 &  Refinement Phase Version 2\\
% 		\hline
%         $40 \%$ (minutes) &$[0.5,~0.6]$ &$[5.8,~5.9]$ &$[4.6,~4.7]$   \\
% 		\hline
% 		$89 \%$ (minutes) &$[3,~4]$& $[9,~10]$ & $[5.5,~6.5]$  \\
% 		\hline
% 	\end{tabular}
% 	}
% 	\label{tab:runtimes}
% \end{table}

% \usepackage{multirow}
% \usepackage{graphicx}
\begin{table*}[htbp]
\centering
\caption{Run-times}
\label{tab:runtimes}
\resizebox{\textwidth}{!}{%
\begin{tabular}{|c|c|c|c|c|c|}
\hline
Maximum Link Utilization &
  Initial Phase (TFA-DRR) &
  Refinement Phase Version 1 &
  Refinement Phase Version 2 &
  Post-Process (iPLP) &
  Post-Process (PLP) \\ \hline
40 \% &
  {[}0.5,0.6{]} (minutes) &
  {[}5.8,5.9{]} (minutes) &
  {[}4.6,4.7{]} (minutes) &
    {\begin{tabular}[c]{@{}c@{}}{[}6.4,6.5{]} (seconds)\\ \end{tabular}} &
  {\begin{tabular}[c]{@{}c@{}}{[}5.7,5.75{]} (seconds)\\ \end{tabular}} \\ \cline{1-4}
  % \multirow{2}{*}{\begin{tabular}[c]{@{}c@{}}{[}6.4,6.5{]}\\ (s)\end{tabular}} &
  % \multirow{2}{*}{\begin{tabular}[c]{@{}c@{}}{[}5.7,5.75{]}\\ (s)\end{tabular}} \\ \cline{1-4}
89 \% &
  {[}3,4{]} (minutes) &
  {[}9,10{]} (minutes) &
  {[}5.5,6.5{]} (minutes) &
   &
   \\ \hline
\end{tabular}%
}
\end{table*}

In Fig.~\ref{fig:numFull}~(a),(b) we compare delay bounds of 1) and 2) to our full method; TFA diverges for class $2$ at link utilization of $89 \%$ whereas our full method remains stable at all link utilizations below $100 \%$. In Fig.~\ref{fig:numFull}~(c), we compare 3) to our full method. Our PLP improvements reduce the delay bounds as well as the run-times: when we use PLP per-aggregate, we solve fewer PLPs. In Fig.~\ref{fig:numFull}~(d), we compare 4) to our full method. We show numerically that when delay bounds are small, iPLP captures a non-convex part of DRR strict service curves and brings an improvement compared to PLP. Experimentally, the improvements increase with the link utilization.

Note that infinite bounds might only be obtained in networks with cyclic dependencies that are highly loaded (as we assume local stability), e.g., in the industrial network we tested, when the maximum link utilization reaches $89 \%$, TFA-DRR provides infinite bounds (see Fig.~\ref{fig:numFull}). Obtaining finite bounds implies that the network is truly stable; whereas, obtaining infinite bounds does not necessarily imply that the network is unstable, and  the network might or might not be truly stable.

We use MATLAB on a $2.6$ GHz $6$-Core Intel Core i$7$ computer; thus, we have $6$ workers to implement our two parallel versions of Section \ref{sec:refinePhase}.  We provide the $95\%$ confidence interval for the run-times of Initial phase (TFA) and  two parallel implementations of the refinement phase at two different maximum link utilization; we run the program $10$ times. When the maximum link utilization is $40 \%$ and $89 \%$, $95\%$ confidence intervals of run-times are provided in Table~\ref{tab:runtimes}. Regarding the post-process phase, we compare run-times of PLP and iPLP for the flows with the longest path. We run the program $100$ time and give the $95\%$ confidence intervals  in Table~\ref{tab:runtimes}.

Our full method, PLP-DRR, clearly dominates the bounds and the stability region compared to our initial phase, TFA-DRR (see Fig.~\ref{fig:numFull}), but, comes with longer run times. As observed in this industrial-sized network, PLP-DRR converges relatively fast, and  bounds are obtained in several minutes; also, we experimentally observe that even in a very large network (a ring-shaped topology with $20$ nodes and $3$ DRR classes), our method PLP-DRR converges in $6$ hours; note that these run-times are obtained using  a 2.6 GHz 6-Core Intel Core i7 laptop, and one should expect a fraction of these run-times in a server or a more powerful computing setup. Nevertheless, as long as PLP-DRR is feasible and has a finite run-time, say hours, the extended running time of PLP-DRR compared to TFA-DRR is not a  concern. This is because typically finding delay bounds in a network setting is an offline problem, hence run times that are in the order of several hours are acceptable. Moreover, our bounds in the refinement phase are always valid even before convergence, hence, one can always stop iterating in the refinement phase with respect to a run-time budget and obtain delay bounds that are at least as good as those of the initial phase. Lastly,  we experimentally observe that iPLP comes at a negligible cost compared to PLP (see Table~\ref{tab:runtimes}) even when the network becomes very large.

\section{Conclusion}
\label{sec:conc}
% \vspace{-0.2cm}
We solved the problem of how to combine DRR strict service curves and the network analysis of PLP in order to obtain worst-case delay bounds in time-sensitive networks. Our method is guaranteed to find delay bounds that are at least as good as the state-of-the-art, and we found very significant improvements on the industrial network under study. It is based on a generic shared memory execution model, implementations of which can differ by the scheduling of the individual operations in the refinement phase. We proved that all implementations produce the same bounds. We proposed two concrete implementations and found that the latter performs faster. It will be interesting to study other concrete implementations that aim at reducing computing time.

\bibliographystyle{IEEEtran}
\vspace{-0.05in}
\bibliography{ref2}
% \bibliography{ref}
%\appendix
%\input{nc}
\begin{appendices}
% !TeX root = mainTON.tex
\section{Detailed Background on DRR Strict Service Curves \cite{drr_rtas,drr_ton}} \label{app:drr}
Here we present more background on DRR strict service curves of \cite{drr_ton}, using our notations, that enables a reader to implement what we use in the paper.
\vspace{-0.2cm}
\subsection{Degraded Operational Mode} \label{app:drrDM}
\vspace{-0.1cm}

Here we present Corollary 2 of \cite{drr_ton} that presents a convex strict service curve for DRR, in  degraded operational mode:

Let $v$ be a node that, shared by $n$ classes,  uses DRR, as explained in Section \ref{sec:drr}, with quantum $Q_c^v$ for class $c$. The node offers a strict service curve $B^v$ to the aggregate of the $n$ classes. For any class $c$, $\mdelta_c$ is the maximum residual deficit  defined by $\mdelta_c = \lmax_c - \epsilon$  where $\lmax_c$ is an upper bound on the packet size of flows of class $c$ at node $v$ and $\epsilon$ is the smallest unit of information seen by the scheduler (e.g., one bit, one byte, one 32-bit word, ...).
    
Then, for every $c$, $v$ offers to class $c$ a strict service curve $\betacdm$ given by $\betacdm(t)=\gammacon \lp B^v(t) \rp$ with $\gammacon = 	 \max\lp\beta_{\rmax_c,\tmax_c}, \beta_{\rmin_c,\tmin_c}\rp$,
 %    \begin{align}
	% 	\label{eq:gammacon}
	% 	\gammacon &= 	 \max\lp\beta_{\rmax_c,\tmax_c}, \beta_{\rmin_c,\tmin_c}\rp
	% \end{align}
 % where 
 $\rmax_c = \frac{Q^v_c}{Q^v_\tot}$, $\tmax_c = \sum_{c',c' \neq c}\lp Q^v_{c'} + \mdelta_{c'}  +\frac{Q^v_{c'}}{Q^v_c} \mdelta_c  \rp$, $\rmin_c = \frac{Q^v_c - \mdelta_c}{Q^v_\tot - \mdelta_c}$, $\tmin_c = \sum_{c',c' \neq c}\lp Q^v_{c'} + \mdelta_{c'} \rp$, and $Q^v_{\tot} =\sum_{c} Q^v_c$.
 %    \begin{align}
	% 	\label{eq:gammacon}
	% 	\gammacon &= 	 \max\lp\beta_{\rmax_c,\tmax_c}, \beta_{\rmin_c,\tmin_c}\rp
	% 	\\
	% 	\rmax_c &= \frac{Q^v_c}{Q^v_\tot},~  \tmax_c = \sum_{c',c' \neq c}\lp Q^v_{c'} + \mdelta_{c'}  +\frac{Q^v_{c'}}{Q^v_c} \mdelta_c  \rp  \\
	% 	\rmin_c &= \frac{Q^v_c - \mdelta_c}{Q^v_\tot - \mdelta_c},~  \tmin_c = \sum_{c',c' \neq c}\lp Q^v_{c'} + \mdelta_{c'} \rp
	% 	% \\
	% 	% 	\label{eq:Qtot}
	% 	% Q^v_{\tot} &=\sum_{c} Q^v_c
	% \end{align}
 % and $Q^v_{\tot} =\sum_{c} Q^v_c$.
 % In \eqref{eq:gammacon}, $\beta_{R,T}$ is a rate-latency function defined in Table~\ref{tab:not}.

For the non-convex strict service curve, $\beta_c^{\scriptsize\textrm{nc},v}$, we have
\begin{align}
    q_c^v = Q^v_c - \mdelta_c \mand T_1 = \tmin_c \label{eq:t1} 
\end{align}
% \begin{align}
%     q_c^v &= Q^v_c - \mdelta_c \label{eq:qc} \\
%     T_1 &= \tmin_c \label{eq:t1}
% \end{align}

\subsection{Non-Degraded Operational Mode} \label{app:drrNDM}

Here we present a new formulation of Corollary~4 of \cite{drr_ton}. We slightly generalize Corollary~4 of \cite{drr_ton}, using Lemma 2 of \cite{drr_ton}, such that it enables us to take into account any available output arrival curves. Specifically, Corollary~4 of \cite{drr_ton} is an application of this new formulation where we replace  $\alpha^*_{c'}$ with $\alpha_{c'} \oslash \oldservice_{c'}$, which is an output arrival curve for class $c'$,  in \eqref{eq:betamJ}. Note that $\oslash$ is the min-plus deconvolution defined by $(y \oslash y')(t) = \sup_{s \geq 0} \{ y(t+s) - y'(s)\}$ for two non-decreasing functions $y$ and $y'$ 
%, and max-plus convolution   is defined by $(f \bar{\otimes} g)(t) = \sup_{0 \leq s \leq t}  \{ f(t-s) + g(s)\}$
\cite{le_boudec_network_2001,Changbook,bouillard_deterministic_2018}.

% \jylb{Need to explain why we giev a new formulation. Was there a bug in [18] ? Give more hints on what we do (we slightly generalize [18] by allowing the output arrival curve to be anything. }
% To obtain this, we combine Corollary~4 of \cite{drr_ton} and Lemma 2 of \cite{drr_ton}.  Also, we present Algorithm 2 of \cite{drr_ton}, adopted with the new formulation: 

Let $v$ be a node with the assumptions in Section \ref{app:drrDM}. Also, assume that every class $c$  has an output arrival curve $\alpha^*_c $ and a strict service curve $\oldservice_c $, and let $N_c = \{c_1,c_2,\ldots,c_n\}\setminus\{c\}$, and for any $J \subseteq N_c$, let $\bar{J} =  N_c \setminus J$.
% 	\anne{Why does $N_i$ not contain $i$?}
	Then, for every class $c$, a new strict service curve $\newservice_c$ is given by
	\begin{equation}
		\label{eq:betamJ}
		\newservice_c = \max \lp \oldservice_c, \max_{J \subseteq N_c}{\gammacon}^J \circ \lb B^v - \sum_{c' \in \bar{J}} \alpha^*_{c'}  \rb_{\uparrow}^+  \rp
	\end{equation}
	with  ${\gammacon}^J = 	 \max\lp\beta_{{\rmax_c }^J,{\tmax_c }^J} , \beta_{{\rmin_c}^J,{\tmin_c}^J}\rp$, 		
	% \begin{align}
	% 	\label{eq:gammaMaxrateJ}
	% 			{\gammacon}^J &= 	 \max\lp\beta_{{\rmax_c }^J,{\tmax_c }^J} , \beta_{{\rmin_c}^J,{\tmin_c}^J}\rp\\
	% {\rmax_c }^J&= \frac{Q^v_c}{Q^{J,c}_\tot},~ {\tmax_c }^J= \sum_{c' \in J}\lp Q^v_{c'} + \mdelta_{c'}   +\frac{Q^v_{c'}}{Q^v_c} \mdelta_c  \rp \label{eq:rateMax}\\
	% {\rmin_c}^J &= \frac{Q^v_c - \mdelta_c}{Q^{J,c}_\tot - \mdelta_c},~  {\tmin_c }^J= \sum_{c' \in J}\lp Q^v_{c'}  + \mdelta_{c'}  \rp	\label{eq:rateMin}	\\
	% 		\label{eq:QtotJ}
	% 	Q_{\tot}^{J,c} &=Q_c  + \sum_{c'\in J} Q_{c'} 
	% \end{align}
 % 	\begin{align}
	% 	\label{eq:gammaMaxrateJ}
	% 			{\gammacon}^J &= 	 \max\lp\beta_{{\rmax_c }^J,{\tmax_c }^J} , \beta_{{\rmin_c}^J,{\tmin_c}^J}\rp
	% \end{align}
  $ {\rmax_c }^J = \frac{Q^v_c}{Q^{J,c}_\tot}$, ${\tmax_c }^J= \sum_{c' \in J}\lp Q^v_{c'} + \mdelta_{c'}   +\frac{Q^v_{c'}}{Q^v_c} \mdelta_c  \rp$, $ {\rmin_c}^J = \frac{Q^v_c - \mdelta_c}{Q^{J,c}_\tot - \mdelta_c}$, ${\tmin_c }^J= \sum_{c' \in J}\lp Q^v_{c'}  + \mdelta_{c'}  \rp$, $Q_{\tot}^{J,c} =Q_c  + \sum_{c'\in J} Q_{c'} $.
% 	\anne{More precisely, shouldn't we have $Q^J_{tot} + Q_i$ instead of $Q^J_{tot}$?} \hossein{Correct and Fixed \eqref{eq:QtotJ}} \anne{But now $Q^J_{tot}$ depends on $i$... it could be better to have $i$ in $J$ and define $T^{max, J}_i = \sum_{j\in J-\{i\}} ...$ or $\phi_{i, j}(x) = x$.}
	
In \eqref{eq:betamJ}, $\lb. \rb ^+_{\uparrow}$ is the non-decreasing and non-negative  closure: The non-decreasing and non-negative closure  $\lb y \rb^+_{\uparrow}$ of a function $y:\mathbb{R}^+ \to \mathbb{R^+} \cup \{+\infty\}$ is the smallest non-negative, non-decreasing  function that upper bounds $y$. Also, $\circ$ is the composition of functions.  $\beta_{R,T}$ is a rate-latency function defined in Table~\ref{tab:not};  note that a rate-latency function, as defined in Table~\ref{tab:not}, has a rate expressed in bit/s, and  a latency expressed in seconds, however, rate and latencies defined above are respectively unitless and in bits. This is because $\beta_{{\rmax_c }^J,{\tmax_c }^J}$ and $\beta_{{\rmin_c}^J,{\tmin_c}^J}$ are later composed by a function expressed in bit/s, hence the final results also  are in bit/s and seconds.
% \bleu{ note that  rate and latencies defined in \eqref{eq:rateMax} and \eqref{eq:rateMin} are expressed with no unit and in bit, respectively, which differ from rate-latency function, as defined in Table~\ref{tab:not}.}

% \bleu{note that rate-latency function, as defined in Table~\ref{tab:not}, has a rate expressed in say bit/s, and  a latency expressed in say seconds, however, rate and latencies defined in \eqref{eq:rateMax} and \eqref{eq:rateMin} are expressed with no unit and in bit, respectively.}

% Note that Corollary~4 of \cite{drr_ton} is an application of this new formulation where we replace  $\alpha^*_{c'}$ with $\alpha_{c'} \oslash \oldservice_{c'}$ in \eqref{eq:betamJ}.

% We use the following in the theorem: The non-decreasing and non-negative closure  $\lb y \rb^+_{\uparrow}$ of a function $y:\mathbb{R}^+ \to \mathbb{R^+} \cup \{+\infty\}$ is the smallest non-negative, non-decreasing  function that upper bounds $y$.

The essence of \eqref{eq:betamJ} is as follows. Equation~\eqref{eq:betamJ} gives new strict service curves $\newservice_c$ for every flow $c$; they are derived from already available strict service curves $\oldservice_c $ and from output arrival curves of classes $\alpha^*_c$; this enables us to improve any collection of  strict service curves that are already obtained. 

Let $\Pi^{\scriptsize\textrm{convex}}_v : \mathscr{F}^{2n}\to \mathscr{F}^n$ be the mapping at server $v$ that maps $\lp \oldservice_1, \oldservice_2, \ldots, \oldservice_n \rp$ using $\lp \alpha^*_1, \alpha^*_2, \ldots, \alpha^*_n \rp$  to $\lp \newservice_1, \newservice_2, \ldots, \newservice_n \rp$ as in  \eqref{eq:betamJ}. Then, an iterative scheme can be defined as in Algorithm \ref{alg:iterScheme}. Note that Algorithm 3 is exactly the same as Algorithm 2 of \cite{drr_ton}, adopted with the new formulation that is presented in \eqref{eq:betamJ}.
% \jylb{Same comment: what is the nature of Algo 3 ? Is it new ? is it the same as Algo 2 of [18] ? Is it a more compact formulation ?}
% \vspace{-0.4cm}
\begin{algorithm}[htbp]
    \SetKwInOut{IV}{Local Variables}

    \KwResult{Collection of strict service curves $\lp \beta^v_1, \ldots, \beta^v_n\rp$}
    \IV{Collections of strict service curves $\lp\oldservice_1, ..., \oldservice_n\rp$ and $\lp\newservice_1, ..., \newservice_n\rp$}

	\For{$c \leftarrow 1$ to $n$}
	    {
	    $\oldservice_c \gets \betacdm$ \;
	    }
	
	\While{Stopping criteria not reached}
	{
	    \For{$c \gets 1$ to $n$}
	    {
	        $\alpha^*_c \gets \alpha_c\oslash \oldservice_c$ \;
	    }
	    $\newservice \gets \Pi^{\scriptsize\textrm{convex}}_v \lp \alpha^*,  \oldservice \rp $\;
	}
	
	$\beta^v \gets \newservice$\;
    \textbf{return} $\lp \beta^v_1, \ldots, \beta^v_n\rp$
    \caption{$\drri_v(\alpha_1, \ldots, \alpha_n)$}
    \label{alg:iterScheme}
\end{algorithm}

	% \vspace{-0.4cm}
It is shown in \cite{drr_ton} that  for every class $c$ , $\oldservice_c $ and $\newservice_c $ are strict service curves for class $c$ and $\oldservice_c \leq \newservice_c$, i.e., an increasing sequence of strict service curves is obtained for every class. Also, this sequence is a guaranteed simple convergence, starting from any valid initial  strict service curves. Note that the computed strict service curves at each iteration are valid hence can be used to derive valid delay bounds; this means the iterative scheme can be stopped at any iteration. For example, the iterative scheme can be stopped when the delay bounds of all classes decrease insignificantly. The scheme requires being initialized by strict service curves. We use $\betacdm$, obtained in Section \ref{app:drrDM} for the initial strict service curves at lines 1-2.

When every class $c$ has a token-bucket arrival curve at the output, say $\gamma_{r_c,b_c^v}$, and a known strict service curve $\beta^v_c$, $\drrb_v \lp \beta^v ,b^v\rp$ is the function that implements \eqref{eq:betamJ} (i.e., $\Pi^{\scriptsize\textrm{convex}}_v$ defined above) and returns an improved collection of strict service curves for all classes at node $v$.

\section{Detailed Background on PLP \cite{plp}} \label{app:plp}
Here we present more background on PLP of \cite{plp}, using our notations, which enables a reader to implement what we use in the paper. Specifically, we summarize linear programs used by PLP. For the rest of the section, a reader is invited to recall the definitions of Section~\ref{sec:2imp}. Readers who know reference \cite{plp} can map our notations to those of \cite{plp} 
and vice-versa as follows: For time variables,  map $\mathbf{t}_{(j,k)}$ of \cite{plp} to $\tvar{v}{k}$  by mapping $j$ to $v$. For process variables,  map $\mathbf{F}^j_i\mathbf{t}_{(j',k)}$ of  \cite{plp} to $\fvar{v}{f}{v'}{k}$ by mapping $j$ to $v$, $i$ to $f$, and $j'$ to $v'$.
% More details can be found in \cite[Sections 4 and 6]{plp}: Map $\tvar{v}{k}$ to $\mathbf{t}_{(j,k)}$ of \cite{plp} by mapping $j$ to $v$, and map $\fvar{v}{f}{u}{k}$ to $\mathbf{F}^j_i\mathbf{t}_{(h,k)}$ of  \cite{plp} by mapping $j$ to $v$, $u $ to $h$, and $i$ to $f$.

Note that  linear programs of \cite{plp} contains some constraints obtained from the Single Flow Analysis (SFA) delay bounds \cite{bouillard_deterministic_2018}. However, in practice, such constraints have no or negligible effects as SFA bounds are often dominated by those of TFA. Hence, in this paper, we do not use SFA constraints.

\subsection{$\plpd_{f, c}$: For an End-to-end Delay Bound of a Flow} \label{app:plpd}

We summarize linear programs used by PLP (see  \cite[Section 4]{plp} for more details). The goal of $\plpd_{f, c}( \beta_c, d_c,\zcutc)$ is to find a valid end-to-end delay bound for a flow of interest $f$ that belongs to a class $c$. We assume a sub-tree of the cut network where the root is the sink server of the flow of interest $f$ (i.e., the last server is traversed by  flow $f$). Recall that we add an artificial node, node $v_0$, that is the successor of the root. We call $\cV_c^f$ the set of output ports in this sub-tree. We assume that the burstiness of flows at cuts is given, i.e., $\zcutc$. Also, for each node $v$, a convex, piece-wise linear service curve (i.e., $\beta_c$) and a per-node delay bounds (i.e., $d_c$) are provided.

% In the rest of this section, we assume a sub-tree of the cut network where the root is the sink server of the flow of interest $f$ (i.e., the last server is traversed by  flow $f$). We call $\cV_c^f$ the set of output ports in this sub-tree.

\subsubsection{Constraints}

In the constraints we define below, let server $u = \suc \lp v \rp$ be the successor of server $v$. Also, we denote flows by $g$,  not to be confused by the flow of interest $f$. 
%For every flow $g$, recall that $v_g$ is the source server of flow $g$.

$\sbullet[.75]$ Time Constraints:
\begin{enumerate}
    \item[-] $\forall v \in \cV^f_c,~ \forall k \in [0,~\depth\lp v\rp - 1], \tvar{v}{k} \geq \tvar{v}{k+1} $;
    \item[-] $\forall v \in \cV^f_c,~ \forall k \in [0,~\depth\lp v\rp], \tvar{v}{k} \leq \tvar{u}{k} $.
\end{enumerate}

$\sbullet[.75]$ FIFO Constraints:
\begin{enumerate}
    \item[-] $\forall v \in \cV^f_c,~ \forall k \in [0,~\depth\lp u\rp], \forall g \in \textrm{In}_c(v), \fvar{v}{g}{v}{k} = \fvar{u}{g}{u}{k} $.
\end{enumerate}

$\sbullet[.75]$ Service Curve Constraints: Recall that $\beta_c^v$ is  piece-wise linear convex (i.e., $\beta_c^v = \max_p \beta_{R_p^v, T_p^v})$), and also recall $\mathbf{At}_u^v = \sum_{g \in \textrm{In}(v)} \fvar{u}{g}{u}{\scriptsize\depth(u)}$ and  $\mathbf{At}_v^v =\sum_{g \in \textrm{In}(v)} \fvar{v}{g}{v}{\scriptsize\depth(v)}$, where $u = \suc(v)$. Then, $\forall v \in \cV^f_c$,
\begin{enumerate}
    \item[-] $\mathbf{At}_u^v - \mathbf{At}_v^v 
        %\sum_{f \in \textrm{In}(v)}\lp \fvar{u}{f}{u}{\scriptsize\depth(u)} - \fvar{v}{f}{v}{\scriptsize\depth(v)} \rp
        \geq 0$;
    \item[-] $\forall p,~\mathbf{At}_u^v - \mathbf{At}_v^v 
        %\sum_{f \in \textrm{In}(v)}\lp \fvar{u}{f}{u}{\scriptsize\depth(u)} - \fvar{v}{f}{v}{\scriptsize\depth(v)} \rp
        \geq R^v_p\lp \tvar{u}{\scriptsize\depth(u)} - \tvar{v}{\scriptsize\depth(v)}  - T^v_p  \rp$.
\end{enumerate}

$\sbullet[.75]$ Per-node Delay Bound Constraints:
\begin{enumerate}
    \item[-] $\forall v \in \cV^f_c,~ \forall k \in [0,~\depth\lp u\rp],~ \tvar{u}{k} - \tvar{v}{k} \leq d_c^v$.
\end{enumerate}

$\sbullet[.75]$ Arrival Curve Constraints: For each flow $g$ of class $c$, recall that $v_g$ is the source server of flow $g$. Define $\bar{b}_g$, the burstiness bound of flow $g$, as follows: If flow $g$ is a fresh flow, its arrival curve is a token-bucket arrival curve $\gamma_{r_g,b_g}$, as defined in Section \ref{sec:sysmodel}, and thus $\bar{b}_g = b_g$. Else, flow $g$ is a cut flow, its arrival curve is a token-bucket arrival curve $\gamma_{r_g,\bar{b}_g }$ where $\bar{b}_g$ is obtained from $\zcutc$.

% or a cut flow. For the former, its arrival curve is defined by a token-bucket arrival curve $\gamma_{r_g,b_g}$ and $\bar{b}_g = b_g$ and for the latter, the burstiness of its arrival curve, $\bar{b}_g$, is obtained from $\zcutc$. 
% For every flow $g$,
\begin{enumerate}
    \item[-] For every flow $g$, $\forall 0 \leq k < k' \leq \depth(v_g),~ \fvar{v_g}{g}{v_g}{k} - \fvar{v_g}{g}{v_g}{k'} \leq \bar{b}_g + r_g\lp \tvar{v_g}{k} - \tvar{v_g}{k'} \rp$.
\end{enumerate}

$\sbullet[.75]$ Shaping Constraints:
% Recall that $\lmax_c$ is a upper bound on the packet sizes of flows of class $c$ at node $v$. 
For every $v \in \cV_c^f$ and every edge $e = (v,u)$, let $F_{v,u}$ be the set of flows of class $c$, carried by the edge $e = (v,u)$. Then,
\begin{enumerate}
    \item[-] $\forall 0 \leq k < k' \leq \depth(u),~ \sum_{g \in F_{v,u}} \lp \fvar{u}{g}{u}{k} - \fvar{u}{g}{u}{k'} \rp  \leq \lmax_c + R^v\lp \tvar{u}{k} - \tvar{u}{k'} \rp$.
\end{enumerate}

$\sbullet[.75]$ Monotonicity Constraints:

Recall that for each flow $g$ of class $c$, $v_g$ is the source server of flow $g$. Then, for every flow $g$,
\begin{enumerate}
    \item[-] $ \forall k \in [0, \depth\lp v_g\rp- 1],~ \fvar{v_g}{g}{v_g}{k} \geq  \fvar{v_g}{g}{v_g}{k+1}$.
\end{enumerate}

\subsubsection{Objective} The Objective is $\max \lp \tvar{v_0}{0} - \tvar{v_f}{0} \rp  $.
% $\sbullet[.75]$ The Objective is $\max \lp \tvar{v_0}{0} - \tvar{v_f}{0} \rp  $
% \begin{enumerate}
%     \item[-] $\max \lp \tvar{v_0}{0} - \tvar{v_f}{0} \rp  $
% \end{enumerate}

\subsection{$\plpb_{f, c}$: For a Backlog Bound of a  Flow}\label{app:plpb}

The goal of $\plpb_{f, c}( \beta_c, d_c,\zcutc)$ is to find a valid backlog bound for a flow of interest $f$ that belongs to a class $c$.  By \cite[Theorem 4]{plp}, the objective of this linear program, is a bound on the burstiness of flow $f$ at the output of node $v$. 

\subsubsection{Constraints}
It contains all the constraints  of $\plpd_{f, c}$ (see Section~\ref{app:plpd}), with the following changes:

First, for shaping constraints, we should remove the flow of interest $f$ from $F_{v,u}$. Specifically, we should replace  $F_{v,u}$ by $F_{v,u} \setminus \{f\}$, i.e.,
the set of flows of class $c$, carried by the edge $e = (v,u)$, excluding flow of interest $f$.

Second, we add the following constraints: Recall that $v_f$ is the source of flow of interest $f$.
\begin{enumerate}
    \item[-] $\forall k \in [0, \depth \lp v_f \rp ],~ \fvar{v_f}{f}{v_0}{0} - \fvar{v_f}{f}{v_f}{k} \leq b_f + r_f\lp \tvar{v_0}{0} - \tvar{v_f}{k} \rp$.
\end{enumerate}

\subsubsection{Objective} The Objective is $\max \lp \fvar{v_f}{f}{v_0}{0} -  \fvar{v_0}{f}{v_0}{0} \rp  $.
% Third, we use the following objective:

% $\sbullet[.75]$ The Objective:
% \begin{enumerate}
%     \item[-] $\max \lp \fvar{v_f}{f}{v_0}{0} -  \fvar{v_0}{f}{v_0}{0} \rp  $
% \end{enumerate}

% $\sbullet[.75]$ Service Curve Constraints:
% \begin{enumerate}
%     \item[-] $\forall v \in \cV_c,~ 
%     \sum_{g \in \textrm{In}_c(v)}\fvar{u}{g}{u}{\scriptsize\depth(u)} \geq  \sum_{g \in \textrm{In}_c(v)} \fvar{v}{g}{v}{\scriptsize \depth(v)}$
%     \item[-] $\forall v \in \cV_c,~ 
%     \sum_{g \in \textrm{In}_c(v)}\fvar{u}{g}{u}{\scriptsize\depth(u)} \geq  \sum_{g \in \textrm{In}_c(v)} \fvar{v}{g}{v}{\scriptsize \depth(v)}$    
    
% \end{enumerate}

% For every node $v$ and every $f$ at $v$, define $\fvar{v}{f}{v}{k}$ with $k \in \{0, \ldots, \depth(v) \}$, where  $\fvar{v}{f}{v}{k}$ is a variable for the cumulative arrival function of flow $f$ at the input of node $v$ at time $\tvar{v}{k}$.  It can be mapped to $\mathbf{F}^j_i\mathbf{t}_{(j,k)}$ of  \cite{plp} by mapping $j$ to $v$, $i$ to $f$.

\subsection{$\plpcut_c$: For Bounds on The Burstiness of Flows at Cuts}\label{app:plpcut}

$\plpcut_c(\beta_c, d_c)$ finds  valid bounds on the burstiness for flows of class $c$ at cuts (i.e., $\zcutc$). We assume for each node $v$, a convex, piece-wise linear service curve (i.e., $\beta_c$) and a per-node delay bounds (i.e., $d_c$) are provided. 

Let $\fcutc$ be the set of cut flows of class $c$. For each flow $f \in \fcutc$, we define a variable $\xvar{f}$ that represents the burstiness of the arrival curve of flow $f$ at its source. $\plpcut_c$ is constructed as follows: For each cut flow $f \in \fcutc$, a fresh set of time and process variables is defined, and all constraints of $\plpb_{f,c}$ (see Section \ref{app:plpb}) are added to the set of constraints of $\plpcut_c$; The common variables between constraints of different cut flows are only $\xvar{f}$ variables. Then, $\plpcut_c$ maximizes the sum of all $\xvar{f}$ variables; \cite[Theorems~7,8]{plp} shows that if the solution is bounded, the values of $\xvar{f}$ in the solution, are valid bounds on the burstiness of cut flows.  

\subsubsection{Constraints}
% We  define the constraints of $\plpcut_c$ as follows:
$\sbullet[.75]$ For each fresh flow $f \in \fcutc$ (i.e., flow $f$ has an arrival curve $\gamma_{b_f,r_f}$), we add $\xvar{f} \leq b_f  $;
% \begin{enumerate}
%     \item[-] $\xvar{f} \leq b_f  $
% \end{enumerate}

$\sbullet[.75]$ For each transit flow $f \in \fcutc$ (i.e., a flow that is not a fresh flow), we first define a fresh set of time and process variables, say $\mathbf{t}$ and $\mathbf{Ft}$, and we add all constraints of $\plpb_{f,c}$, defined in Section \ref{app:plpb}. Also, the objective of $\plpb_{f,c}$ is added as an constraint for $\xvar{f}$: $\xvar{f} \leq \lp \fvar{v_f}{f}{v_0}{0} -  \fvar{v_0}{f}{v_0}{0} \rp$. Note that for arrival curve constraints of a cut flow $g \in \fcutc$, the burstiness $\bar{b}_g$ is replaced by the variable $\xvar{g}$.

\subsubsection{Objective} The Objective is $\max \sum_{f\in \fcutc} \xvar{f}  $.
% $\sbullet[.75]$ The Objective:
% \begin{enumerate}
%     \item[-] $\max \sum_{f\in \fcutc} \xvar{f}  $
% \end{enumerate}

\end{appendices}

\end{document}
%	\pagenumbering{arabic}